\documentclass[a4paper,UKenglish]{llncs}
 
\usepackage{microtype}
\usepackage{amsmath,amssymb} 
\usepackage{latexsym} 
\usepackage{epic,eepic}
\usepackage{graphics}
\usepackage{graphicx}
\usepackage{epsfig}
\usepackage{xypic}
\usepackage{bussproofs}
\usepackage{stmaryrd}

\newcommand{\solve}{Solve}
\newcommand{\matching}{Matching}

\newcommand{\lrules}{\Delta_\Lambda}
\newcommand{\grules}{\Delta_{G}}

\newcommand{\At}{\mathit{Atoms}}
\newcommand{\bottom}{\bot}
\newcommand{\A}{\mathcal{A}}
\newcommand{\Q}{\mathcal{Q}}
\def \R {\mathcal{R}}

\newcommand{\leqTF}{\sqsubseteq}
\newcommand{\geqTF}{\sqsupseteq}
\newcommand{\LTA}{LTA}
\newcommand{\cond}{\Leftarrow}
\newcommand{\Pred}{{\cal P}}
\newcommand{\Lang}{\mathcal{L}}
\newcommand{\eval}{eval}
\newcommand{\desc}{\R^*}

\newcommand{\F}{{\cal F}}
\newcommand{\X}{{\cal X}}

\newcommand{\FI}{{\cal F}_{\bullet}}
\newcommand{\FNI}{{\cal F}_{\circ}}
\newcommand{\TF}{{\cal T(F)}}
\newcommand{\TFX}{{\cal T(F, X)}}
\newcommand{\TFNIX}{{\cal T(\FNI, X)}}
\newcommand{\TFIX}{{\cal T(\FI, X)}}
\newcommand{\TFI}{{\cal T(\FI)}}

\newcommand{\TFAL}{{\cal T}(\F,\mathit{Atoms}(\Lambda))}

\newcommand{\TFQ}{{\cal T(F \cup Q)}}

\newcommand{\rw}{\rightarrow}

\newcommand{\simp}{\leadsto}

\newcommand{\nr}{E}
\newcommand{\comp}{{\cal C}}
\newcommand{\aapprox}{{\cal A}_{\R, \nr}}
\newcommand{\aaex}{{\cal A}_{\R}}

\def\build#1_#2^#3{\mathrel{
    \mathop{\kern 0pt#1}\limits_{#2}^{#3}}}

\def\CM#1{\build\hbox to 10mm {\rightarrowfill}_{}^{CM#1}}

\newcommand{\rws}[3]{\mathrel{{\build{\rightarrow}_{#1}^{#2}}\mskip-2mu_{#3}}}
\newcommand{\pos}{{\cal P}os}

\newcommand{\var}{{\cal V}ar}

\newcommand{\rwR}{\rws{}{}{\R}}
\newcommand{\rwA}{\rws{}{}{\A}}
\newfont{\amstoto}{msbm10}
\newcommand{\NN}{\mbox{\amstoto\char'116}}
\newcommand{\sep}{\; | \;}

\newcommand{\aut}{\langle \F, \Q, \Q_f, \Delta \rangle} 

\newcommand{\opA}{op^{\#}}

\newcommand{\plA}{+^{\#}}
\newcommand{\msA}{-^{\#}}
\newcommand{\FpA}{\FI^{\#}}
\newcommand{\evA}{eval^{\#}}

\title{Tree Regular Model Checking for\\ Lattice-Based Automata}

\author{Thomas Genet \inst{1} \and Tristan Le Gall \inst{2} \and Axel Legay \inst{1} \and Val\'erie Murat \inst{1}}

\institute{INRIA/IRISA, Rennes \and CEA, LIST, Centre de recherche de Saclay}

\begin{document}
\maketitle
\sloppy

\begin{abstract}
{\it Tree Regular Model Checking} (TRMC) is the name of a family of
techniques for analyzing infinite-state systems in which states are
represented by terms, and sets of states by Tree Automata (TA). The
central problem in TRMC is to decide whether a set of bad states is
reachable. The problem of computing a TA representing (an
over-approximation of) the set of reachable states is undecidable, but
efficient solutions based on completion or iteration of tree
transducers exist.

Unfortunately, the TRMC framework is unable to efficiently capture
both the complex structure of a system and of some of its features. As
an example, for JAVA programs, the structure of a term is mainly
exploited to capture the structure of a state of the system. On the
counter part, integers of the java programs have to be encoded with
Peano numbers, which means that any algebraic operation is potentially
represented by thousands of applications of rewriting rules.

In this paper, we propose Lattice Tree Automata (LTAs), an extended
version of tree automata whose leaves are equipped with lattices. LTAs
allow us to represent possibly infinite sets of interpreted terms. Such
terms are capable to represent complex domains and related operations
in an efficient manner. We also extend classical Boolean operations to
LTAs. Finally, as a major contribution, we introduce a new
completion-based algorithm for computing the possibly infinite set of
reachable interpreted terms in a finite amount of time.
\end{abstract}

\section{Introduction}

Infinite-state models are often used to avoid potentially artificial
assumptions on data structures and architectures, e.g. an artificial
bound on the size of a stack or on the value of an integer
variable. At the heart of most of the techniques that have been
proposed for exploring infinite state spaces, is a symbolic
representation that can finitely represent infinite sets of states.

In early work on the subject, this representation was domain specific,
for example linear constraints for sets of real
vectors\,\cite{WB98}. For several years now, the idea that a generic
automata-based representation for sets of states could be used in many
settings has gained ground starting with finite-word
automata\,\cite{BLW03,BLW04,KMMPS97,AJNd03}, and then moving to the
more general setting of Tree Regular Model Checking
(TRMC)\,\cite{AJMd02,BT02,ALRd05}. In TRMC, states are represented by
trees, set of states by tree automata, and behavior of the system by
rewriting rules or tree transducers. Contrary to specific approaches,
TRMC is generic and expressive enough to describe a broad class of
communication protocols\,\cite{ALRd05}, various C
programs\,\cite{BHRV06b} with complex data structures, multi-threaded
programs, and even cryptographic
protocols~\cite{GenetK-CADE00,avispa-site}. Any Tree Regular Model
Checking approach is equipped with an acceleration algorithm to
compute possibly infinite sets of states in a finite amount of
time. Among such algorithms, one finds completion by equational
abstraction\,\cite{GR10} that computes successive automata obtained by
application of the rewriting rules, and merge intermediary states
according to an equivalence relation to enforce the termination of the
process.

In \cite{BoichutGJL-RTA07}, the authors proposed an exact translation
of the semantic of the Java Virtual Machine to tree automata and
rewriting rules. This translation permits to analyze java programs
with classical Tree Regular Model checkers. One of the major
difficulties of this encoding is to capture and handle the two-side
infinite dimension that can arise in Java programs. Indeed, in such
models, infinite behaviors may be due to unbounded calls to method and
object creation, or simply because the program is manipulating
unbounded data such as integer variables. While multiple infinite
behaviors can be over-approximated with completion and equational
abstraction\,\cite{GR10}, their combinations may require the
use of artificially large-size structures. As an example in
\cite{BoichutGJL-RTA07}, the structure of a configuration is
represented in a very concise manner as the structure of terms is
mainly designed to efficiently capture program counters, stacks,
.... On the other hand, integers and their related operations have to
be encoded in Peano arithmetic, which has an exponential impact on the
size of automata representing sets of states as well as on the
computation process. As an example, the addition of $x$ to $y$
requires the application of $x$ rewriting rules.

A solution to the above problem would be to follow the solution of
Kaplan\,\cite{DBLP:conf/rta/KaplanC89}, and represent integers in
bases greater or equal to $2$, and the operations between them in the
alphabet of the term directly. In such a case, the term could be
interpreted and returns directly the result of the operation without
applying any rewriting rule. The study of new Tree Regular Model
Checking approaches for such interpreted terms is the main objective
of this paper. Our first contribution is the definition of {\em
  Lattice Tree Automata} ($\LTA$), a new class of tree automata that
is capable of representing possibly infinite sets of interpreted
terms. Roughly speaking, $\LTA$ are classical Tree Automata whose
leaves may be equipped with lattice elements to abstract possibly
infinite sets of values. Nodes of $\LTA$ can either be defined on an
uninterpreted alphabet, or represent lattice operations, which will
allows us to interpreted possibly infinite sets of terms in a finite
amount of time. We also propose a study of all the classical
automata-based operations for $\LTA$. The model of $\LTA$ is not
closed under determinization. In such case, the best that can be done
is to propose an over-approximation of the resulting automaton through
abstract interpretation. As a third contribution, we propose a new
acceleration algorithm to compute the set of reachable states of
systems whose states are encoded with interpreted terms and sets of
states with $\LTA$. Our algorithm extends the classical completion
approach by considering conditional term rewriting systems for
lattices. We show that dealing with such conditions requires to merge
existing completion algorithm with a solver for abstract domains. We
also propose a new type of equational abstraction for lattices, which
allows us to enforce termination in a finite amount of time. Finally,
we show that our algorithm is correct in the sense that it computes an
over-approximation of the set of reachable states. This latter
property is only guaranted providing that each completion step is
followed by an evaluation operation. This operation, which relies on a
widening operator, add terms that may be lost during the completion
step. Finally, we briefly describe how our solution can drastically
improve the encoding of Java programs in a TRMC environment.

\paragraph{\bf Related Work} 

This work is inspired by \cite{DBLP:conf/sas/GallJ07}, where the
authors proposed to use finite-word lattice automata to solve the
Regular Model Checking problem. Our major differences are that (1) we
work with trees, (2) we propose a more general acceleration
algorithm, and (3) we do consider operations on lattices while they
only consider to label traces with lattices without permitting to
combine them. Some Regular Model Checking approaches can be find in
\cite{ADR07,BLW03,AHDHR08,Bouajjani:2006:ART:1706621.1706712}. However,
none of them can capture the two infinite-dimensions of complex
systems in an efficient manner. Other models, like modal
automata\,\cite{BFJLLT11} or data trees\,\cite{FS11,GMW10}, consider
infinite alphabets, but do not exploit the lattice structure as in our
work. Lattice~(-valued) automata~\cite{KL07}, whose transitions are
labelled by lattice elements, map words over a finite alphabet to a
lattice value. Similar automata may define fuzzy tree
languages~\cite{Esik:2007:FTA:1243507.1243664}. Other verification of
particular classes of properties of Java programs with interpreted
terms can be found in \cite{OBEG10}.

\section{Backgrounds}

%%In this section, we introduce some definitions and concepts that will
%%be used throughout the rest of the paper.

\paragraph{Rewriting Systems and Tree Automata.}

Let $\F$ be a finite set of functional symbols, where each symbol is
associated with an arity, and let $\X$ be a countable set of {\em
  variables}. $\TFX$ denotes the set of {\em terms} and $\TF$ denotes
the set of {\em ground terms} (terms without variables). The set of
variables of a term $t$ is denoted by $\var(t)$. The set of functional
symbols of arity $n$ is denoted by $\F^{n}$.  A {\em position} $p$ for
a term $t$ is a word over $\NN$. The empty sequence $\varepsilon$
denotes the top-most position. We denote by $Pos(t)$ the set of
position of a term $t$.  If $p \in \pos(t)$, then $t|_p$ denotes the
subterm of $t$ at position $p$ and $t[s]_p$ denotes the term obtained
by replacement of the subterm $t|_p$ at position $p$ by the term $s$.

A {\em Term Rewriting System} (TRS) $\R$ is a set of {\em rewrite
  rules} $l \rw r$, where $l, r \in \TFX$, and $\var(l) \supseteq
\var(r)$.  A rewrite rule $l \rw r$ is {\em left-linear} if each
variable of $l$ occurs only once in $l$.  A TRS $\R$ is left-linear if
every rewrite rule $l \rw r$ of $\R$ is left-linear.

We now define Tree Automata ($TA$ for short) that are used to
recognize possibly infinite sets of terms. Let $\Q$ be a finite set of
symbols of arity $0$, called {\em states}, such that $\Q \cap \F=
\emptyset$. The set of {\em configurations} is denoted by $\TFQ$. A
         {\em transition} is a rewrite rule $c \rightarrow q$, where
         $c$ is a configuration and $q$ is a state.  A transition is
         {\em normalized} when $c = f(q_1, \ldots, q_n)$, $f \in \F$
         is of arity $n$, and $q_1, \ldots, q_n \in \Q$.  A bottom-up
         nondeterministic finite tree automaton (tree automaton for
         short) over the alphabet $\F$ is a tuple $\A= \langle \Q,\F,
         \Q_F,\Delta \rangle$, where $\Q_F \subseteq \Q$ is the set of
         final states, $\Delta$ is a set of normalized transitions.

The transitive and reflexive {\em rewriting relation} on $\TFQ$
induced by $\Delta$ is denoted by $\rightarrow_{\A}^*$. The tree
language recognized by $\A$ in a state $q$ is $\Lang{}(\A,q) = \{t \in
\TF \sep t \rw^*_\A q \}$. We define $\Lang{}(\A) = \bigcup_{q \in
  \Q_F} \Lang{}(\A, q)$.
  
\paragraph{Lattices, atomic lattices, Galois connections.}

A partially ordered set ($\Lambda,\sqsubseteq$) is a lattice if it
admits a \textit{smallest element} $\bot$ and a \textit{greatest
  element} $\top$, and if any finite set of elements $X \subseteq
\Lambda$ admits a \textit{greatest lower bound (glb)} $\sqcap X$ and a
\textit{least upper bound (lub)} $\sqcup X$ . A lattice is complete if
the \textit{glb} and \textit{lub} operators are defined for all
possibly infinite subset of $\Lambda$.  An element $x$ of a lattice
($\Lambda,\sqsubseteq$) is an atom if it is minimal, \textit{i.e.}
$\bot \sqsubset x \wedge \forall y \in \Lambda : \bot \sqsubset y
\sqsubseteq x \Rightarrow y = x$. The set of atoms of $\Lambda$ is
denoted by $Atoms(\Lambda)$. A lattice ($\Lambda,\sqsubseteq$) is
atomic if all element $x \in \Lambda$ where $x \neq \bot$ is the least
upper bound of atoms, \textit{i.e.} $x = \bigsqcup\lbrace a \vert a
\in Atoms(\Lambda) \wedge a \sqsubseteq x\rbrace$.

Considered two lattices $(C,\sqsubseteq_C)$ (the concrete domain) and
$(A,\sqsubseteq_A)$ (the abstract domain). We say that there is a
\emph{Galois connection} between the two lattices if there are two
monotonic functions $\alpha : C \rightarrow A$ and $\gamma : A
\rightarrow C$ such that : $\forall x \in C, y \in A$, $ \alpha(x)
\sqsubseteq_A y$ if and only if $x \sqsubseteq_C \gamma(y)$.  As an
example, sets of integers $(2^\mathbb{Z},\subseteq)$ can be abstracted
by the atomic lattice $(\Lambda,\sqsubseteq)$ of intervals, whose
bounds belong to $\mathbb{Z} \cup \{ -\infty,+\infty\}$) and whose
atoms are of the form $[x,x]$, for each $x \in \mathbb{Z}$. Any
operation $op$ defined on a concrete domain $C$ can be lifted to an
operation $\opA$ on the corresponding abstract domain $A$, thanks to
the Galois connection.

\section{Lattice Tree Automata}
\label{sec:LTA}

In this section, we first explain how to add elements of a concrete
domain into terms, which has been defined in
\cite{DBLP:conf/rta/KaplanC89}, and how to derive an abstract domain
from a concrete one. Then we propose a new type of tree automata
recognizing terms with elements of a lattice and study its
properties.\\

%%%%%%%%%%%%%%%%%%%%%%%%%%%%%%%%%%%%%%%%%%%%%%%%%%%%%%%%%%%%%%%%%%%%%%%%%%%%%%%%%%%%%%%%%%%%%%%%%%

\subsection{Discussion}
We first discuss the reason for which we chose to consider tree
automata with leaves that are labelled by elements of an atomic
lattice. We remind that the main goal of this work is to extend the
TRMC approach to tree automata that represent sets of interpreted
terms. We may assume that the interpreted terms of a given set are
similar to each other, for example $\{f(1),f(2),f(3),f(4)\}$. We can
encode naively this set of terms by a tree automaton with the
transitions~: $1 \to q, 2 \to q, 3 \to q, 4 \to q, f(q) \to q_f$. This
naive encoding is quite inefficient, and we would prefer to label the
leaves of the tree not by integers, but by a set of integers. The new
tree automata has only two transitions~: $\{1,2,3,4\}\to q, f(q) \to
q_f$.

This is the reason why we considered the notion of $\LTA$~: In there,
sets of integers is just a particular lattice. By considering tree
automata with a generic lattice, we can also improve the efficiency of
the approach. General sets of integers are indeed hard to handle, and
we often only need an over-approximation of the set of reachable
states. That is why we prefer to label the leaves of the tree by
elements of an abstract lattice $\Lambda$ such as the lattice of
intervals. The Galois connection ensures that the concrete operations
(e.g. $+, \times$) on integers have an abstract semantics, and that
the approximations are sound.

In order to simplify the notations, we did not emphasize in this paper
the abstract interpretation aspects. For example, when we say that
``the concrete domain is $\mathcal{D} = \mathbb{N}$, the abstract
domain is $(\Lambda,\sqsubseteq)$'', it really means that the concrete
lattice is $(2^\mathbb{N},\subseteq)$ and that there is a Galois
connection with $(\Lambda,\sqsubseteq)$. In the examples, we apply
implicitely the concretization function, wich is the identity (if the
abstract lattice is the lattice of intervals). We can also define the
$\LTA$ even when there is no Galois connection between the concrete
lattice and the abstract one. In this case, the function $\evA$ must
be defined so that we still have over-approximation of the concrete
operations.

There are two reasons why we consider only atomic abstract lattices,
and why the language of an $\LTA$ is defined on tarms built with the
atoms rather that with any elements of the lattice. The first one is
that we are mostly interested in representing sets of integers. Since
the atoms are the integers, the semantics of a lambda transition is to
recognize a set of integers. The other reason is a technical one~: It
ensures that when we transform a $\LTA$ according to a partition, we
do not change the recognized language since the set of atoms are
preserved by this transformation.

%%%%%%%%%%%%%%%%%%%%%%%%%%%%%%%%%%%%%%%%%%%%%%%%%%%%%%%%%%%%%%%%%%%%%%%%%%%%%%%%%%%%%%%%%%%%%%%%%%

\subsection{Interpreted Symbols and Evaluation}

In what follows, elements of a concrete and possibly infinite domain
$\mathcal{D}$ will be represented by a set of {\em interpreted}
symbols $\FI$. The set of symbols is now denoted by $\F = \FNI \cup
\FI$, where $\FNI$ is the set of {\em passive} (uninterpreted)
symbols.  The set of {\em interpreted} symbols $\FI$ is composed of
elements of $\mathcal{D}$ (\textit{i.e} $\mathcal{D} \subseteq \FI$)
whose arity is 0, and is also composed of some predefined operations
$f$ : $\mathcal{D}^{n} \rightarrow \mathcal{D}$, where $f \in
\F^{n}$. For example, if $\mathcal{D} = \mathbb{N}$, then $\FI$ can be
$\mathbb{N} \cup \lbrace +, -, *\rbrace$.  Passive symbols can be seen
as usual non-interpreted functional operators, and interpreted symbols
stand for {\em built-in} operations on the domain $\mathcal{D}$.

The set $\TFI$ of terms built on $\FI$ can be evaluated by using an
eval function $\eval: \mathcal{T}(\FI) \rightarrow \mathcal{D}$.  The
purpose of $\eval$ is to simplify a term using the built-in operations
of the domain $\mathcal{D}$.  The $\eval$ function naturaly extends to
$\TF$ in the following way: $\eval(f(t_1, \ldots,t_n)= f(\eval(t_1),
\ldots, \eval(t_n))$ if $f\in \FNI$ or $\exists i=1\ldots n: t_i
\not\in\TFI$. Otherwise, $f(t_1,\ldots,t_n) \in \TFI$ and the
evaluation returns an element of $\mathcal{D}$.

To deal with infinite alphabets (e.g. $\mathbb{R}$ or $\mathbb{N}$),
we propose to replace the concrete domain $\mathcal{D}$ by an abstract
one $\Lambda$, linked to $\mathcal{D}$ by a Galois
connection. Moreover, we assume that $(\Lambda,\sqsubseteq)$ is an
\emph{atomic lattice} and that the built-in symbols are $\sqcup$ and
$\sqcap$, which arity is 2, and other symbols corresponding to the
abstraction of $\FI$.

Let $OP$ be the set of operations $op$ defined on $\mathcal{D}$, and
$OP^{\#}$ the set of corresponding operations $\opA$ defined on
$\Lambda$, we have that $\FI = \mathcal{D} \ \cup \ OP$, and the corresponding
abstract set is defined by $\FpA = \Lambda \ \cup \ OP^{\#} \ \cup
\ \lbrace \sqcup, \sqcap\rbrace$.  For example, let $I$ be the set of
intervals with bounds belonging to $\mathbb{Z} \cup \{
-\infty,+\infty\}$. The set $\FI = \mathbb{Z} \cup \lbrace +,-
\rbrace$ can be abstracted by $\FpA = I \cup \lbrace \plA, \msA,
\sqcup, \sqcap \rbrace$.  Terms containing some operators extended to
the abstract domain have to be evaluated, like explained in section
3.2 for the concrete domain. $\evA : \FpA \to \Lambda$ is the best
approximation of $\eval$ w.r.t. the Galois connection.

\begin{example}[$\evA$ function]
For the lattice of intervals on $\mathbb{Z}$, we have that:
\begin{itemize}
\item $\evA(i) = i$ for any interval $i$,
\item For any $f \in \lbrace \plA, \msA,
\sqcup, \sqcap \rbrace$  $\evA(f(i_1,i_2))$ is defined, given $\evA(i_1) = [a,b]$
  and $\evA(i_2) = [c,d]$, by:
 $\evA([a,b] \sqcup [c,d])= [min(a,c), max(b,d)]$,
 $\evA([a,b] \sqcap [c,d])= [max(a,c), min(b,d)]$ if $max(a,c)\leq min(b,d)$, else $\evA([a,b] \sqcap [c,d])= \bot$,
$ \evA([a,b]\plA [c,d])=[a + c , b + d]$,
$ \evA([a,b]\msA [c,d])=[a -d , b - c]$.

\end{itemize}
  
\end{example}

\subsection{The Lattice Tree Automata Model}

Lattice tree automata are extended tree automata recognizing terms
defined on $\FNI \cup \FI^{\#}$.

\begin{definition}[lattice tree automaton]
  A bottom-up non-deterministic finite tree automaton with lattice
  (lattice tree automaton for short, $\LTA$) is a tuple $\A= \langle
  \F = \FNI \cup \FI^{\#}$, $\Q, \Q_f,\Delta \rangle$, where $\F$ is a
  set of passive and interpreted symbols, $\Q$ and $\Q_{f}$ a set of
  state, $\Q_f \subseteq \Q$, and $\Delta$ is a set of normalized
  transitions.
\end{definition}

\noindent
The set of \textit{lambda transitions} is defined by $\lrules=\{
\lambda \rw q \sep \lambda \rw q \in \Delta \: \wedge \: \lambda \neq \bot \wedge \: \lambda \in
\Lambda\}$.  The set of \emph{ground transitions} is the set of other
transitions of the automaton, and is formally defined by $\grules=\{
f(q_1,\dots,q_n) \rw q \sep f(q_1,\dots,q_n)\rw q \in \Delta \: \wedge
\: q,q_1,\dots,q_n \in Q\}$.\\

We extend the partial ordering $\sqsubseteq$ (on $\Lambda$) on $\TF$:
\begin{definition}
\label{def:sqsubterm}
  Given $s,t\in \TF$, $s \leqTF t$ iff (1) $s \leqTF t$~(if both $s$
  and $t$ belong to $\Lambda$), (2) $eval(s) \leqTF eval(t)$~(if both
  $s$ and $t$ belong to $\TFI$), (3) $s = t$~(if both $s$ and $t$
  belong to $\FNI^{0}$), or (4) $s=f(s_1, \ldots, s_n)$, $t=f(t_1,
  \ldots, t_n)$, $f\in\FNI^{n}$ and $s_1 \leqTF t_1 \wedge \ldots
  \wedge s_n \leqTF t_n$.
\end{definition}

\begin{example}
$f(g(a, [1,5]) \leqTF f(g(a, [0,8])$, and $h([0,4] +^{\#} [2,6])
  \leqTF h([1,3] +^{\#} [1, 9])$.
\end{example}

In what follows we will omit $\#$ when it is clear from the
context. We now define the transition relation induced by an $\LTA$.
The difference with $TA$ is that a term $t$ is recognized by an $\LTA$
if $eval(t)$ can be reduced in the $\LTA$.
\begin{definition}[$t_{1} \rw_{\A} t_{2}$ for lattice tree automata]
\label{def:run}
Let $t_{1}, t_{2} \in \mathcal{T}(\F\cup \mathcal{Q})$.

$t_{1} \rw_{\A} t_{2}$ iff, for any position $p \in pos(t_{1})$ :

\begin{itemize}
\item if $t_{1}|_p \in \TFI$, there is a transition $\lambda
  \rightarrow q \in \mathcal{A}$ such that $eval(t_{1}|_p) \sqsubseteq
  \lambda$ and $t_{2} = t_{1}[q]_p$
\item if $t_{1}|_p = a$ where $a \in \FNI$, there is a transition $a
  \rightarrow q \in \mathcal{A}$ such that $t_{2} = t_{1}[q]_p$.
\item if $t_{1}|_p = f(s_{1}, \ldots, s_{n})$ where $f \in \F^{n}$ and
  $s_{1}, \ldots s_{n} \in \mathcal{T}(\F\cup \mathcal{Q})$, $\exists
  s'_{i} \in \mathcal{T}(\F\cup \mathcal{Q})$ such that $s_{i}
  \rw_{\A} s'_{i}$ and $t_{2} = t_{1}[f(s_{1}, \ldots, s_{i-1},
    s'_{i}, s_{i+1}, \ldots, s_{n})]_p$.

\end{itemize}

\end{definition}
$ \rw^{*}_{\A}$ is the reflexive transitive closure of
$\rw_{\A}$. There is a run from $t_1$ to $t_2$ if $t_{1} \rw^{*}_{\A}
t_{2}$.

The set $\TFAL$ denotes the set of ground terms built over
$(\F\setminus\Lambda)\cup Atoms(\Lambda)$.  Tree automata with lattice
recognize a tree language over $\TFAL$.

\begin{definition}[Recognized language]
\label{def:rec}
  The tree language recognized by $\A$ in a state $q$ is $\Lang(\A,q)
  = \{t \in \TFAL \, | \, \exists \ t' \ such \ that \ t \sqsubseteq
  t'\ and \ t' \rw^{*}_{\A} q\}$.  The language recognized by $\A$ is
  $\Lang(\A) = \bigcup_{q \in \Q_f} \Lang(\A, q)$.
  
\end{definition}

\begin{example}[Run, recognized language]
Let $\A= \langle \F = \FNI \cup \FI^{\#}$, $\Q, \Q_f,\Delta \rangle$
be an $\LTA$ where $\Delta = \lbrace [0, 4] \rightarrow q_{1}, f(q_{1})
\rightarrow q_{2} \rbrace$ and final state $q_2$.  We have: $f([1, 4])
\rightarrow^{*} q_{2}$ and $f([0,2]) \rightarrow^{*} q_{2}$, and the
recognized langage of $\A$ is given by $\Lang(\A, q_{2}) = \{
f([0,0]), f([1,1]), \ldots, f([4,4]) \}$.
\end{example}

\subsection{Operations on $\LTA$}

Most of the algorithms for Boolean operations on $\LTA$ are
straightforward adaptations of those defined on $TA$ (see
\cite{tata2007}).

$\LTA$ are closed by union and intersection, and we shortly explain
how these two operations $\cup$ and $\cap$ can be performed on two
$\LTA$s $\A= \langle \F, \Q, \Q_f,\Delta \rangle$ and $\A'= \langle
\F, \Q', \Q_f',\Delta' \rangle$~:
\begin{itemize}
\item $\A \cup \A' = \langle \F, \Q \cup \Q', \Q_f \cup \Q_f',\Delta
  \cup \Delta'\rangle$ assuming that the sets $\Q$ and $\Q'$ are
  disjoint.
\item $\A \cap \A'$ is recognized by the $\LTA$ $\A \cap \A'= \langle \F,
  \Q \times Q', \Q_f \times \Q_f',\Delta_\cap \rangle$ where the
  transitions of $\Delta_\cap$ are defined by the rules:
\[
%\begin{array}{|@{\hspace*{1ex}}c@{\hspace*{1ex}}|@{\hspace*{1ex}}c@{\hspace*{1ex}}|}
%\hline
\begin{array}{c}
\lambda \to q \in \Delta \quad \lambda' \to q' \in \Delta' \quad \lambda \sqcap \lambda' \neq \bot \\
\hline
 \lambda \sqcap \lambda' \to (q,q')
\end{array}
%&
\]
\noindent and
\[
\begin{array}{c}
f(q_1,\dots,q_n) \to q \in \Delta \quad f(q_1',\dots,q_n') \to q' \in \Delta' \\
\hline
 f((q_1,q'_1),\dots,(q_n,q'_n)) \to (q,q')
\end{array} %\\
%\hline
%\end{array}
\]
\end{itemize}

Assuming that the $\LTA$ is deterministic, the complement automaton is
obtained by complementing the set of final states.  To decide if the
language described by an $\LTA$ is empty or not, it suffices to observe
that an $\LTA$ accepts at least one tree if and only if there is an
reachable final state.  A reduced automaton is an automaton without
inaccessible state. The language recognized by a reduced automaton is
empty if and only if the set of final states is empty.  
As a first step we thus have to reduce the $\LTA$, that is to remove the set of unreachable states.

Let us recall the reduction algorithm:
\begin{tabbing}
\textbf{Reduction Algorithm}\\
 \textbf{input:} $\LTA$ $\A = \langle \F, \Q, \Q_{f},\Delta \rangle$\\
 \textbf{begin} \= \\
 \> \textit{Marked}:=$\emptyset$ \\
 \> /* Marked is the set of accessible states */\\
 \> repeat \= \\
 \> \> if \= $\exists a \in \F^{0} = \FNI^{0} \cup \FI^{\#^{0}}$ such that $a \rightarrow q \in \Delta$\\
 \> \> \> or \= $\exists f \in \F^{n} = \FNI^{n} \cup \FI^{\#^{n}}$ such that $f(q_{1}, \ldots, q_{n}) \rightarrow q \in \Delta$ \\
 \> \> \> \> where $q_{1}, \ldots, q_{n} \in Marked$ \\
\> \> then $Marked := Marked \cup \{q\}$ \\
\> until no state can be added to \textit{Marked} \\
\> $\Q_{r} := Marked$ \\
\> $\Q_{r_{f}} := \Q_{f} \cap Marked$ \\
\> $\Delta_{r} := \{f(q_{1}, \ldots, q_{n}) \rightarrow q \in \Delta | q, q_{1}, \ldots, q_{n} \in Marked\}$ \\
\> \textbf{output:} Reduced $\LTA$ $\A_{r} = \langle \F, \Q_{r}, \Q_{r_{f}},\Delta_{r} \rangle$ \\
\textbf{end}
\end{tabbing}

Then, let $\A$ be an $\LTA$ and $\A_{r} = \langle \F, \Q_{r}, \Q_{r_{f}},\Delta_{r} \rangle$ the corresponding reduced $\LTA$, $\Lang(\A)$ is empty iff $\Q_{r_{f}} = \emptyset$.

Let $\A_{1}$, $\A_{2}$ be two $\LTA$. We have $\Lang(\A_{1})\subseteq\Lang(\A_{2}) \Leftrightarrow \Lang(\A_{1}\cap \overline{\A_{2}}) = \emptyset$.\\

Complementation and inclusion requires an input deterministic
$\LTA$. However, by adapting the proof of finite-word lattice automata
given in \cite{DBLP:conf/sas/GallJ07}, one can show that $\LTA$ are
not closed under determinization. In the next section, we propose an algorithm
that computes an over-approximation deterministic automaton for any
given $LTA$. This algorithm, which extends the one of
\cite{DBLP:conf/sas/GallJ07}, relies on a partition function that can
be refined to make the overapproximation more precise.

\subsection{Determinization}

As we shall now see, an $\LTA$ $\A = \langle \F, \Q, \Q_f,\Delta \rangle$ is \textit{deterministic} if there is no transition $f(q_{1}, ..., q_{n}) \rightarrow q$, $f(q_{1}, ..., q_{n}) \rightarrow q'$ in $\Delta$ 
such that $q \neq q'$, where $f \in \mathcal{F}_{n}$, and no transition $\lambda_{1} \rightarrow q$, $\lambda_{2} \rightarrow q'$ such that $q \neq q'$ and $\lambda_{1} \sqcap \lambda_{2} \neq \bot$, where $\lambda_{1}, \lambda_{2} \in \Lambda$.
As an example, if  $\Delta = \lbrace [1,3] \rightarrow q_{1}, [2,5] \rightarrow q_{2}\rbrace$, then we have that $\A$ is not deterministic.

Determinizing an $\LTA$ requires complementation on elements on lattice. Indeed, consider the $\LTA$ $\A$ having the following transitions $[-3,2] \rightarrow q_{1}$ and $[1,6] \rightarrow q_{2}$. The deterministic $\LTA$ corresponding to $\A$ should have the following transitions:  $[-3,1[ \rightarrow q_{1}$, $[1,2] \rightarrow \{q_{1}, q_{2} \}$ and $]2,6] \rightarrow q_{2}$. To produce those transitions, we have to compute $[-3,2]\sqcap [1,6] = [1,2]$, and then $[-3,2]\setminus [1,2]$ and $[1,6] \setminus [1,2]$.
Unfortunately, there are lattices that are not closed under complementation. 
 As a consequence, determinization of an $\LTA$ does not preserve the recognized language.

 The solution proposed in \cite{DBLP:conf/sas/GallJ07} for word automata
 is to use a finite partition of the lattice $\Lambda$, which commands when two transitions
should be merged using the $lub$ operator. The fusion of transitions may induce an over-approximation controlled by the fineness of the partition.\\

\textbf{Partitioned \textit{LTA}.}
$\Pi$ is a \textit{partition} of an atomic lattice $\Lambda$ if $\Pi \subseteq 2^\Lambda$ and $\forall \pi_{1}, \pi_{2} \in \Pi$, $\pi_{1} \sqcap \pi_{2} = \bot$, 
and $\forall a \in \At(\Lambda), \exists \pi \in \Pi: a \sqsubseteq \pi$.
As an example, if $\Lambda$ is the lattice of intervals, we can have a partition $\Pi = \lbrace ]-\infty, 0[, [0,0], ]0,+\infty[\rbrace$.

\begin{definition}[Partitioned lattice tree automaton ($P\LTA$)]
A $P\LTA$ $\mathcal{A}$ is an $\LTA$ $\mathcal{A} = \langle \Pi, \mathcal{Q}, \mathcal{F}, \mathcal{Q}_{f}, \Delta \rangle$ equipped with a partition $\Pi$,
such that for all lambda transitions $\lambda \rightarrow q \in \Delta$, $\exists \pi \in \Pi$ such that $\lambda \sqsubseteq \pi$.

A $P\LTA$ is \textit{merged} if $\lambda_{1} \rightarrow q$, $\lambda_{2} \rightarrow q \in \Delta \wedge \lambda_{1} \sqsubseteq \pi_{1} \wedge \lambda_{2} \sqsubseteq \pi_{2}$
$\implies$ $\pi_{1} \sqcap \pi_{2} = \emptyset$, where $\lambda_{1}, \lambda_{2} \in \Lambda$ and $\pi_{1}, \pi_{2} \in \Pi$.
\end{definition}
For example, if $\Pi = \lbrace ]-\infty, 0[, [0,0], ]0,+\infty[\rbrace$, a $P\LTA$ can have the following transition rules : 
$[-3,-1] \rightarrow q_{1}$, 
$[-5,-2] \rightarrow q_{2}$, 
$[3,4] \rightarrow q_{4}$.
This $P\LTA$ is not merged because of the two lambda transitions  $[-3,-1] \rightarrow q_{1}$ and $[-5,-2] \rightarrow q_{2}$, because $[-3,-1]$ and $[-5,-2]$ are in the same partition. The merged corresponding one will have the following transition : $[-5,-1] \rightarrow q_{1,2}$, instead of the two transitions mentionned before.

Any $\LTA$ $\A$ can be turned into a	$P\LTA$ $\A_{p}$ the following way~: 
Let $\Pi$ be the partition. For any lambda transition $\lambda \rightarrow q \in \A$, if $\exists \pi_{1}, \ldots, \pi_{n} \in \Pi$ such that $\lambda \sqcap \pi_{1} \neq \emptyset, \ldots, \lambda \sqcap \pi_{n} \neq \emptyset$, where $\pi_{1} \neq \ldots \neq \pi_{n}$, the transition $\lambda \rightarrow q$ will be replaced by $n$ transitions $\lambda\sqcap \pi_{1} \rightarrow q, \ldots, \lambda\sqcap\pi_{n} \rightarrow q$ in $\A_{p}$.

\begin{example}
\label{exRaf}
Let $\mathcal{A} = \langle \mathcal{Q}, \mathcal{F}, \mathcal{Q}_{f}, \Delta \rangle$ be an $\LTA$ such that $\Delta = \{  [3,4] \rightarrow q_{1}, [-3, 2] \rightarrow q_{2}, f(q_{1},q_{2}) \rightarrow q_{f} \}$, and $\Pi = \{ ]-\infty,0[, [0,0], ]0,+\infty[ \}$ be a partition. Then the corresponding $P\LTA$ is $\mathcal{A}_{p} = \langle \mathcal{Q}, \mathcal{F}, \mathcal{Q}_{f}, \Delta_{p} \rangle$, where $\Delta_{p}= \{  [3,4] \rightarrow q_{1}, [-3, 0[ \rightarrow q_{2}, [0,0] \rightarrow q_{2}, ]0,2] \rightarrow q_{2}, f(q_{1},q_{2}) \rightarrow q_{f} \}$.
\end{example}

Two lambda transitions $\lambda_{1} \rightarrow q$, $\lambda_{2} \rightarrow q$ of a $P\LTA$ can not be merged if $\lambda_{1}$ and $\lambda_{2}$ belong to different elements of the partition, whereas they might be merged in the opposite case.

\begin{proposition}[Equivalence between $\LTA$ and $P\LTA$]
Given an $\LTA$ $\mathcal{A} = \langle \mathcal{Q}, \mathcal{F}, \mathcal{Q}_{f}, \Delta \rangle$ and a partition $\Pi$, there exists a $P\LTA$ $\mathcal{A}' = \langle \Pi, \mathcal{Q}, \mathcal{F}, \mathcal{Q}_{f}, \Delta ' \rangle$ recognizing the same language.% (this relies on the atomic lattice assumption).
\end{proposition}

\begin{proof}
$\mathcal{A}'$ is obtained from $\mathcal{A}$ by replacing each lambda transition $\lambda \rightarrow q \in \Delta$ by at most $n_{\Pi}$ transitions 
$\lambda_{i} \rightarrow q$ where $\lambda_{i} = \lambda \sqcap \pi_{i}$, $\pi_{i} \in \Pi$, such that $\bigsqcup \lambda_{i} = \lambda$.
\end{proof}

Any $P\LTA$ $\mathcal{A} = \langle \Pi, \mathcal{Q}, \mathcal{F}, \mathcal{Q}_{f}, \Delta \rangle$  can be transformed into a
merged $P\LTA$ $\mathcal{A}_{m} = \langle \Pi, \mathcal{Q}, \mathcal{F}, \mathcal{Q}_{f}, \Delta\textit{m} \rangle$ such that $\mathcal{L}(\mathcal{A}) \subseteq \mathcal{L}(\mathcal{A}_{m})$ by merging transitions as follows :
 $\dfrac{q \in \mathcal{Q} \ \ \ \pi \in \Pi \ \ \ \lambda_{m} = \bigsqcup \lbrace \lambda \sqcap \pi, \lambda \in \Lambda \vert \lambda \rightarrow q \in \Delta \rbrace}
{\lambda_{m} \rightarrow q \in \Delta_{m}}$

\begin{example}

If $\mathcal{A} = \langle \Pi, \mathcal{Q}, \mathcal{F}, \mathcal{Q}_{f}, \Delta \rangle$, where $\Pi = ]-\infty,0[,[0,+\infty$ and $\Delta = \{[0,2] \rightarrow q_{1}, [5,8] \rightarrow q_{2}, [-3,-2] \rightarrow q_{3}, [-4,-1] \rightarrow q_{4}, h(q_{1},q_{2},q_{3},q_{4}) \rightarrow q_{f} \}$, the merged automaton $\mathcal{A}_{m} = \langle \Pi, \mathcal{Q}, \mathcal{F}, \mathcal{Q}_{f}, \Delta\textit{m} \rangle$ corresponding to $\A$ has the following transitions: $\Delta_{m} = \{[0,8] \rightarrow q_{1,2}, [-4,-1] \rightarrow q_{3,4}, h(q_{1,2},q_{1,2},q_{3,4},q_{3,4}) \rightarrow q_{f} \}$.

\end{example}

We are now ready to sketch the determinization algorithm.
The determinization of a $P\LTA$, which transforms a $P\LTA$ $\mathcal{A}$ to a merged Deterministic Partitioned $\LTA$ $\mathcal{A}_{d}$ according to a partition $\Pi$,  mimics the one on usual $TA$. The difference is that two $\lambda$-transitions $\lambda_{1} \rightarrow q_{1}$ and $\lambda_{2} \rightarrow q_{2}$ are merged in $\lambda_{1} \sqcap \lambda_{2} \rightarrow \{ q_{1}, q_{2}\}$ when $\lambda_{1}$ and $\lambda_{2}$ are included in the same element $\pi$ of the partition $\Pi$. Consequently, the resulting automaton recognizes a larger language~: $\mathcal{L}(\mathcal{A}) \subseteq \mathcal{L}(\mathcal{A}_{d})$.% The precise determinization algorithm can be found in Appendix. 
This algorithm produces the best approximation in term of inclusion of languages.

\begin{tabbing}
\textbf{Determinization Algorithm :} \\
\textbf{input:} $P\LTA$ $\mathcal{A}$ = $\langle \Pi, \mathcal{Q}, \mathcal{F}, \mathcal{Q}_{f}, \Delta \rangle$\\
\textbf{begin}  \= \\
\> $\mathcal{Q}_{d}$ := $\emptyset$; $\Delta_{d}$ = $\emptyset$;\\
\> \textbf{for all} \= $\pi \in \Pi$ \textbf{do} \\
\> \> $Trans(\pi)$ := $\lbrace \lambda \rightarrow q \in \Delta \vert \lambda \in \Lambda, \lambda \sqsubseteq \pi \rbrace$;\\
\> \> $s := \lbrace q \in \mathcal{Q} \vert \lambda \rightarrow q \in Trans(\pi)\rbrace$;\\
\> \> $\mathcal{Q}_{d}$ := $\mathcal{Q}_{d} \cup \lbrace s \rbrace$; \\
\> \> $\lambda_{m} := \bigsqcup \lbrace \lambda \vert \lambda \rightarrow q \in Trans(\pi)\rbrace$;\\
\> \> $\Delta_{d}$ := $\Delta_{d} \cup \lbrace \lambda_{m} \rightarrow s \rbrace$;\\
\> \textbf{end for}\\

\medskip

 \> \textbf{repeat}\\

\> \> Let $f \in \mathcal{F}_{n}$, $s_{1}, \ldots, s_{n} \in \mathcal{Q}_{d}$,\\ 
\> \> $s := \lbrace q \in \mathcal{Q} \vert \exists q_{1} \in s_{1}, \ldots, q_{n} \in s_{n}, f(q_{1}, \ldots, q_{n}) \rightarrow q \in \Delta \rbrace$;\\
\> \> $\mathcal{Q}_{d}$ := $\mathcal{Q}_{d} \cup \lbrace s \rbrace$; \\
\> \> $\Delta_{d}$ := $\Delta_{d} \cup \lbrace f(s_{1}, \ldots, s_{n}) \rightarrow s \rbrace$;\\

%\begin{tabbing}
%1 tabulation \= 2 tabulation \= 3 tabulation \\
%  \> 2 tabulation \\
%  \> \> 3 tabulation
%\end{tabbing}

 \> \textbf{until} no more rule can be added to $\Delta_{d}$\\
 $\mathcal{Q}_{df}$ := $\lbrace s \in \mathcal{Q}_{d} \vert s \cap \mathcal{Q}_{f} \neq \emptyset\rbrace$\\
\textbf{output} merged $DP\LTA$ $\mathcal{A}_{d}$ = $\langle \Pi, \mathcal{Q}_{d}, \mathcal{F}, \mathcal{Q}_{df}, \Delta_{d} \rangle$\\
\textbf{end} \\
\end{tabbing}

\begin{example}
\label{ex:part}

Let $\Delta = \lbrace[-3,-1] \rightarrow q_{1}, [-5,-2] \rightarrow q_{2}, [3,4] \rightarrow q_{3}, [-3, 2] \rightarrow q_{4},
 f(q_{1},q_{2}) \rightarrow q_{5}, f(q_{3},q_{4}) \rightarrow q_{6}, f(q_{5},q_{6}) \rightarrow q_{f1}, f(q_{5},q_{6}) \rightarrow q_{f2}\rbrace$,
and $\Pi = \lbrace ]-\infty, 0[, [0,0], ]0,+\infty[\rbrace$

With the determinization algorithm defined above, we obtain this set of transition for the deterministic corresponding $P\LTA$ :
$\Delta_{d} = \lbrace, 
[-5,0[ \rightarrow q_{1,2,4}, ]0,4] \rightarrow q_{3,4}, [0, 0] \rightarrow q_{4}, f(q_{1,2,4},q_{1,2,4}) \rightarrow q_{5},
 f(q_{3,4},q_{3,4}) \rightarrow q_{6}, f(q_{3,4},q_{4}) \rightarrow q_{6}, f(q_{3,4},q_{1,2,4}) \rightarrow q_{6}, f(q_{5},q_{6}) \rightarrow q_{f1,f2}\rbrace$.

\end{example}

\begin{proposition}{Deterministic $P\LTA$ is the best upper-approximation}
\label{DetBest}

Let $\mathcal{A}_{1}$ be a $P\LTA$ and $\mathcal{A}_{2}$ the $P\LTA$ obtained with the determinization algorithm. 
Then $\mathcal{A}_{2}$ is a best upper-approximation of $\mathcal{A}_{1}$ as a merged and deterministic $P\LTA$.

\begin{enumerate}
\item $\mathcal{L}(\mathcal{A}_{1}) \subseteq \mathcal{L}(\mathcal{A}_{2})$
\item For any merged and deteministic $P\LTA$ $\mathcal{A}_{3}$ based on the same partition as $\mathcal{A}_{1}$, 
$\mathcal{L}(\mathcal{A}_{1}) \subseteq \mathcal{L}(\mathcal{A}_{3}) \implies \mathcal{L}(\mathcal{A}_{2}) \subseteq \mathcal{L}(\mathcal{A}_{3})$
\end{enumerate}
 
\end{proposition}

\begin{proof}[Proposition \ref{DetBest}]

$(1)$ Base case : for all lambda transitions of $\mathcal{A}_{1}$ $\lambda \rightarrow q$, let $\pi \in \Pi$ such that $\lambda \sqsubseteq \pi$. Then $Trans(\pi)$ = $\lbrace \lambda \rightarrow q \in \Delta \vert \lambda \in \Lambda, \lambda \sqsubseteq \pi \rbrace$. Then there is a transition $\lambda ' \rightarrow Q$ in $\mathcal{A}_{2}$ such that $\lambda ' = \bigsqcup \lbrace \lambda \vert \lambda \rightarrow q \in Trans(\pi)\rbrace$ and $Q = \lbrace q \vert \lambda \rightarrow q \in Trans(\pi)\rbrace$, so $q \in Q$.\\
induction case : for all non lambda transition of $\mathcal{A}_{1}$ $f(q_{1}, \ldots, q_{n}) \rightarrow q$, there is the corresponding transition $f(Q_{1}, \ldots, Q_{n}) \rightarrow Q$ such that $q \in Q$. We have $q_{1} \in Q_{1}, \ldots, q_{n} \in Q_{n}$ thanks to the induction hypothesis.\\
So $\mathcal{L}(\mathcal{A}_{1}) \subseteq \mathcal{L}(\mathcal{A}_{2})._{\square}$\\

$(2)$ $\mathcal{A}_{1}$ = $\langle \Pi, \mathcal{Q}_{1}, \mathcal{F}, \mathcal{Q}_{f_{1}}, \Delta_{1} \rangle$, $\mathcal{A}_{2}$ = $\langle \Pi, \mathcal{Q}_{2}, \mathcal{F}, \mathcal{Q}_{f_{2}}, \Delta_{2} \rangle$ and $\mathcal{A}_{3}$ = $\langle \Pi, \mathcal{Q}_{3}, \mathcal{F}, \mathcal{Q}_{f_{3}}, \Delta_{3} \rangle$\\

As $\mathcal{L}(\mathcal{A}_{1}) \subseteq \mathcal{L}(\mathcal{A}_{2})$ (1) and $\mathcal{L}(\mathcal{A}_{1}) \subseteq \mathcal{L}(\mathcal{A}_{3})$, let $\mathcal{R}_{1} : \mathcal{Q}_{1} \times \mathcal{Q}_{2}$ and $\mathcal{R}_{2} : \mathcal{Q}_{1} \times \mathcal{Q}_{3}$ be two simulation relations defining these properties as follows.

Let $q_{1} \in \mathcal{Q}_{1}$ and $q_{2} \in \mathcal{Q}_{2}$, $(q_{1}, q_{2}) \in \mathcal{R}_{1}$ iff
\begin{itemize}
\item $\lambda_{1} \rightarrow q_{1}  \in \Delta_{1}, \  \lambda_{2} \rightarrow q_{2}  \in  \Delta_{2}$ and $\lambda_{1} \sqsubseteq \lambda_{2}$, where $\lambda_{1} , \lambda_{2} \in \Lambda$,\\
or\\
$f(q_{i_{1}}, \ldots, q_{i_{n}}) \rightarrow q_{1}  \in  \Delta_{1}, \ f(q_{i_{1}}', \ldots, q_{i_{n}}') \rightarrow q_{2}  \in \Delta_{2}$  and 
$\forall j \in [1,  n], \ (q_{i_{j}}, q_{i_{j}}') \in \mathcal{R}_{1}$, where $f \in \mathcal{F}_{n}$
\item $q_{1} \in \mathcal{Q}_{f_{1}} \Longleftrightarrow q_{2} \in \mathcal{Q}_{f_{2}}$ \\
\end{itemize}

Let $q_{1} \in \mathcal{Q}_{1}$ and $q_{3} \in \mathcal{Q}_{3}$, $(q_{1}, q_{3}) \in \mathcal{R}_{2}$ iff
\begin{itemize}
\item $\lambda_{1} \rightarrow q_{1}  \in \Delta_{1}, \  \lambda_{3} \rightarrow q_{3}  \in  \Delta_{3}$ and $\lambda_{1} \sqsubseteq \lambda_{3}$, where $\lambda_{1} , \lambda_{3} \in \Lambda$,\\
or\\
$f(q_{i_{1}}, \ldots, q_{i_{n}}) \rightarrow q_{1}  \in  \Delta_{1}, \ f(q_{i_{1}}', \ldots, q_{i_{n}}') \rightarrow q_{3}  \in \Delta_{2}$  and 
$\forall j \in [1,  n], \ (q_{i_{j}}, q_{i_{j}}') \in \mathcal{R}_{2}$, where $f \in \mathcal{F}_{n}$
\item $q_{1} \in \mathcal{Q}_{f_{1}} \Longleftrightarrow q_{3} \in \mathcal{Q}_{f_{3}}$ \\
\end{itemize}

Let $\mathcal{R} : \mathcal{Q}_{2} \times \mathcal{Q}_{3}$ be a simulation relation such that $(q_{2}, q_{3}) \in \mathcal{R}$ iff $\exists q_{1}
\in \mathcal{Q}_{1}.(q_{1}, q_{2}) \in \mathcal{R}_{1} \wedge (q_{1}, q_{3}) \in \mathcal{R}_{2}$, where $q_{2} \in \mathcal{Q}_{2}$, $q_{3} \in \mathcal{Q}_{3}$.\\

%As $\mathcal{L}(\mathcal{A}_{1}) \subseteq \mathcal{L}(\mathcal{A}_{2})$ $(1)$, and $\mathcal{L}(\mathcal{A}_{1}) \subseteq \mathcal{L}(\mathcal{A}_{3})$, then $\forall q_{1}$ ---------------------COMPLETER\\

Let $(q_{2}, q_{3}) \in \mathcal{R}$. This means that : 
\begin{itemize}
\item $\lambda_{1} \rightarrow q_{1}  \in \Delta_{1}, \  \lambda_{2} \rightarrow q_{2}  \in  \Delta_{2}, \ \lambda_{3} \rightarrow q_{3}  \in  \Delta_{2}$ and $\lambda_{1} \sqsubseteq \lambda_{2}$ and $\lambda_{1} \sqsubseteq \lambda_{3}$, where $\lambda_{1} , \lambda_{2}, \lambda_{3} \in \Lambda$ (a)\\,
or\\
$f(q_{i_{1}}, \ldots, q_{i_{n}}) \rightarrow q_{1}  \in  \Delta_{1}, \ f(q_{i_{1}}', \ldots, q_{i_{n}}') \rightarrow q_{2}  \in \Delta_{2}, \  f(q_{i_{1}}'', \ldots, q_{i_{n}}'') \rightarrow q_{3}  \in \Delta_{3}$  and 
$\forall j \in [1,  n], \ (q_{i_{j}}, q_{i_{j}}') \in \mathcal{R}_{1}$ and $(q_{i_{j}}, q_{i_{j}}'') \in \mathcal{R}_{2}$, where $f \in \mathcal{F}_{n}$ (b)\\
\item $q_{1} \in \mathcal{Q}_{f_{1}} \Longleftrightarrow q_{2} \in \mathcal{Q}_{f_{2}}$ and $q_{1} \in \mathcal{Q}_{f_{1}} \Longleftrightarrow q_{3} \in \mathcal{Q}_{f_{3}}$ (c),
\end{itemize}
by definition of $\mathcal{R}_{1}$ and $\mathcal{R}_{2}$.\\

(a) Let $\pi \in \Pi$ be the element of the partition such that $\lambda_{1} \sqsubseteq \pi$. Then $Trans(\pi)$ = $\lbrace \lambda \rightarrow q \in \Delta \vert \lambda \in \Lambda, \lambda \sqsubseteq \pi \rbrace$, i.e the set of all the lambda transitions $\lambda \rightarrow q$ in $\Delta_{1}$ such that $\lambda \sqsubseteq \pi$. Of course $\lambda_{1} \sqsubseteq Trans(\pi)$, because $\lambda_{1} \sqsubseteq \pi$.
Then 
$\lambda_{2}$ is the least upper bound of all $\lambda \in \Lambda$ such that $\lambda \rightarrow q \in Trans(\pi)$, i.e
$\lambda_{2} = \bigsqcup \lbrace \lambda \vert \lambda \rightarrow q \in Trans(\pi)\rbrace$, according to the determinization algorithm.

As $\mathcal{A}_{3}$ is deterministic and contains $\mathcal{A}_{1}$, then $\lambda_{3}$ has to contain at least all the $\lambda \in \Lambda$ such that $\lambda \rightarrow q \in \Delta_{1}$ and $\lambda \sqsubseteq \pi$, or else $\mathcal{A}_{3}$ is not deterministic.

So $\lambda_{3} \sqsupseteq \bigsqcup \lbrace \lambda \vert \lambda \rightarrow q \in Trans(\pi)\rbrace$, so $\lambda_{2} \sqsubseteq \lambda_{3}$.\\

(b) We can immediately deduce that $\forall j \in [1,  n], \ (q_{i_{j}}', q_{i_{j}}'') \in \mathcal{R}$ by the definition of $\mathcal{R}$.\\

(c) So $q_{2} \in \mathcal{Q}_{f_{2}} \Longleftrightarrow q_{3} \in \mathcal{Q}_{f_{3}}$\\

And thanks to these properties deduced on $\mathcal{R}:\mathcal{Q}_{1}\times\mathcal{Q}_{2}$, we can deduce that $\mathcal{L}(\mathcal{A}_{2}) \subseteq \mathcal{L}(\mathcal{A}_{3})$.\\

As the least upper bound of two elements of a lattice is the best and unique upper-approximation, this determinization algorithm returns the best upper-approximation.$_{\square}$\\

\end{proof}

\subsection{Minimization}
To define the minimization algorithm, we first have to define a \textit{Refine} recursive algorithm which refines an equivalence relation $P$ on states, according to the $P\LTA$ $\A$.\\

\begin{tabbing}
\textbf{Refine($P,\A$)}  \\
 \textbf{begin}  \= \\
 \> Let $P'$ be a new equivalence relation;\\
 \> For \= all  $(q,q') \in \Q$ such that $qPq'$ do \\
 \> \> IF \= ($\forall f \in \F^{n}$, \\
 \> \> \> $\Delta(f(q_{1}, \ldots, q_{i-1}, q, q_{i+1}, \ldots, q_{n}))P\Delta(f(q_{1}, \ldots, q_{i-1}, q', q_{i+1}, \ldots, q_{n}))$,\\ 
 \> \> \> where $q_{1}, \ldots, q_{i-1}, q_{i+1}, \ldots, q_{n} \in \Q$) \\
 \> \> \> AND ($\forall a \in \FNI^{0}$, $a \rightarrow q \Rightarrow a \rightarrow q'$)\\
 \> \> \> AND \= ($\forall \lambda_{1}, \lambda_{2} \in \Lambda$, $\exists \pi \in \Pi$ \\
 \> \> \> \> such that $\lambda_{1} \rightarrow q \Rightarrow \lambda_{2} \rightarrow q'$ and $\lambda_{1}, \lambda_{2} \in \pi$)\\
 \> \> THEN $qP'q$\\
 \> \> ELSE \= if $P = \{\Q_{1}, \ldots, \Q_{i}, \ldots, \Q_{n}\}$ and $q,q' \in \Q_{i}$\\
 \> \> \> then \= $P := \{\Q_{1}, \ldots, \Q_{i-1}, \Q_{i_{1}}, \Q_{i_{2}}, \Q_{i+1}, \ldots, \Q_{n}\}$;\\
 \> \> \> \> $q \in \Q_{i_{1}}$; $q' \in \Q_{i_{2}}$;\\
 \> \> \> \> Refine($P'$);\\
 \textbf{end}\\
\end{tabbing}

\noindent
We are now ready to define the minimization algorithm of a $P\LTA$
$\A$.\\

%Dire que pareil que classique pour les non lambdas transitions et lambda transitions rien à faire dessus car déjà déterminisé. 
\begin{tabbing}
\textbf{MinimizationAlgorithm($\A$)} \\
 \textbf{input:} \= Determinized $P\LTA$ $\mathcal{A}$ = $\langle \Pi, \mathcal{Q}, \mathcal{F}, \mathcal{Q}_{f}, \Delta \rangle$\\
 \> An equivalence relation $P = \{ \Q_{f}, \Q\setminus \Q_{f}\}$\\
 \textbf{output:} Minimized and determinized $P\LTA$ $\mathcal{A_{m}}$ = $\langle \Pi, \mathcal{Q}_{m}, \mathcal{F}, \mathcal{Q}_{f_{m}}, \Delta_{m} \rangle$\\
 \textbf{begin} \= \\
 \> Refine($P$, $\A$);\\
 \> Set $\Q_{m}$ to the set of equivalence classes of $P$;\\
 \> /* we denote by $[q]$ the equivalence class of state $q$ w.r.t. $P$ */ \\
 \> For \= all $\lambda$-transitions, for all $\lambda_{1}, \lambda_{2} \in \Lambda$,\\
 \> \> if $\lambda_{1} \rightarrow q$, $\lambda_{2} \rightarrow q' \in \Delta$ and $qPq'$\\
 \> \> then $\lambda_{1}\sqcup\lambda_{2} \rightarrow [q,q'] \in \Delta_{m}$;\\
 \> For all other transitions, $\Delta_{m} := \{(f,[q_{1}],...,[q_{n}]) → [f(q_{1},...,q_{n})]\}$;\\
 \> $\mathcal{Q}_{m_{f}} := \{[q] | q \in \Q_{f}\}$;\\
 \textbf{end}\\

%For all (q,q') such that qPq', for all f(qqq)q f(qqq)q' is replacing by one transition   
%remplacer ts ceux qui sont eq
%+ treillis U...\\

\end{tabbing}

%Complexity : $O(knlog(n))$, where $n$ is the number of states of $\Q$ and $k$ the size of $\F$.\\
%
%Proof.\\

\noindent
A \textbf{normalized $P\LTA$} is an $\LTA$ that is a merged,
deterministic and minimized $P\LTA$.

\begin{proposition}{Normalized $P\LTA$ is the best upper-approximation}
Let $\mathcal{A}_{1}$ be a $P\LTA$ and $\mathcal{A}_{2}$ the $P\LTA$
obtained with the minimization algorithm.  Then $\mathcal{A}_{2}$ is a
best upper-approximation of $\mathcal{A}_{1}$ as a normalized $P\LTA$.

\begin{enumerate}
\item $\mathcal{L}(\mathcal{A}_{1}) \subseteq \mathcal{L}(\mathcal{A}_{2})$
\item For any normalized $P\LTA$ $\mathcal{A}_{3}$ based on the same partition as $\mathcal{A}_{1}$, 
$\mathcal{L}(\mathcal{A}_{1}) \subseteq \mathcal{L}(\mathcal{A}_{3}) \implies \mathcal{L}(\mathcal{A}_{2}) \subseteq \mathcal{L}(\mathcal{A}_{3})$
\end{enumerate}
\end{proposition}

\noindent
\textbf{Proof :}\\

Let $P$ be the equivalence relation at the end of the minimization algorithm.\\

$(1)$ Base case : for all lambda transitions of $\mathcal{A}_{1}$
$\lambda \rightarrow q$, there is a transition $\lambda ' \rightarrow
[q]$ in $\mathcal{A}_{2}$ such that $\lambda ' = \bigsqcup \lbrace
\lambda \vert \lambda \rightarrow q' \in \Delta_{1} \wedge
q'Pq\rbrace$.\\ induction case : for all non lambda transitions of
$\mathcal{A}_{1}$ $f(q_{1}, \ldots, q_{n}) \rightarrow q$, there is
the corresponding transition $f([q_{1}], \ldots, [q_{n}]) \rightarrow
[q]$ (where $q \in [q]$, $q_{1} \in [q_{1}], \ldots, q_{n} \in
[q_{n}]$).\\ So $\mathcal{L}(\mathcal{A}_{1}) \subseteq
\mathcal{L}(\mathcal{A}_{2})._{\square}$\\

$(2)$ $\mathcal{A}_{1}$ = $\langle \Pi, \mathcal{Q}_{1}, \mathcal{F}, \mathcal{Q}_{f_{1}}, \Delta_{1} \rangle$, $\mathcal{A}_{2}$ = $\langle \Pi, \mathcal{Q}_{2}, \mathcal{F}, \mathcal{Q}_{f_{2}}, \Delta_{2} \rangle$ and $\mathcal{A}_{3}$ = $\langle \Pi, \mathcal{Q}_{3}, \mathcal{F}, \mathcal{Q}_{f_{3}}, \Delta_{3} \rangle$\\

As $\mathcal{L}(\mathcal{A}_{1}) \subseteq \mathcal{L}(\mathcal{A}_{2})$ (1) and $\mathcal{L}(\mathcal{A}_{1}) \subseteq \mathcal{L}(\mathcal{A}_{3})$, let $\mathcal{R}_{1} : \mathcal{Q}_{1} \times \mathcal{Q}_{2}$ and $\mathcal{R}_{2} : \mathcal{Q}_{1} \times \mathcal{Q}_{3}$ be two simulation relations defining these properties as follows.

Let $q_{1} \in \mathcal{Q}_{1}$ and $q_{2} \in \mathcal{Q}_{2}$, $(q_{1}, q_{2}) \in \mathcal{R}_{1}$ iff
\begin{itemize}
\item $\lambda_{1} \rightarrow q_{1}  \in \Delta_{1}, \  \lambda_{2} \rightarrow q_{2}  \in  \Delta_{2}$ and $\lambda_{1} \sqsubseteq \lambda_{2}$, where $\lambda_{1} , \lambda_{2} \in \Lambda$,\\
or\\
$f(q_{i_{1}}, \ldots, q_{i_{n}}) \rightarrow q_{1}  \in  \Delta_{1}, \ f(q_{i_{1}}', \ldots, q_{i_{n}}') \rightarrow q_{2}  \in \Delta_{2}$  and 
$\forall j \in [1,  n], \ (q_{i_{j}}, q_{i_{j}}') \in \mathcal{R}_{1}$, where $f \in \mathcal{F}_{n}$
\item $q_{1} \in \mathcal{Q}_{f_{1}} \Longleftrightarrow q_{2} \in \mathcal{Q}_{f_{2}}$ \\
\end{itemize}

Let $q_{1} \in \mathcal{Q}_{1}$ and $q_{3} \in \mathcal{Q}_{3}$, $(q_{1}, q_{3}) \in \mathcal{R}_{2}$ iff
\begin{itemize}
\item $\lambda_{1} \rightarrow q_{1}  \in \Delta_{1}, \  \lambda_{3} \rightarrow q_{3}  \in  \Delta_{3}$ and $\lambda_{1} \sqsubseteq \lambda_{3}$, where $\lambda_{1} , \lambda_{3} \in \Lambda$,\\
or\\
$f(q_{i_{1}}, \ldots, q_{i_{n}}) \rightarrow q_{1}  \in  \Delta_{1}, \ f(q_{i_{1}}', \ldots, q_{i_{n}}') \rightarrow q_{3}  \in \Delta_{2}$  and 
$\forall j \in [1,  n], \ (q_{i_{j}}, q_{i_{j}}') \in \mathcal{R}_{2}$, where $f \in \mathcal{F}_{n}$
\item  $q_{1} \in \mathcal{Q}_{f_{1}} \Longleftrightarrow q_{3} \in \mathcal{Q}_{f_{3}}$ \\
\end{itemize}

Let $\mathcal{R} : \mathcal{Q}_{2} \times \mathcal{Q}_{3}$ be a simulation relation such that $(q_{2}, q_{3}) \in \mathcal{R}$ iff $\exists q_{1}
\in \mathcal{Q}_{1}.(q_{1}, q_{2}) \in \mathcal{R}_{1} \wedge (q_{1}, q_{3}) \in \mathcal{R}_{2}$, where $q_{2} \in \mathcal{Q}_{2}$, $q_{3} \in \mathcal{Q}_{3}$.\\

%As $\mathcal{L}(\mathcal{A}_{1}) \subseteq \mathcal{L}(\mathcal{A}_{2})$ $(1)$, and $\mathcal{L}(\mathcal{A}_{1}) \subseteq \mathcal{L}(\mathcal{A}_{3})$, then $\forall q_{1}$ ---------------------COMPLETER\\

Let $(q_{2}, q_{3}) \in \mathcal{R}$. This means that : 
\begin{itemize}
\item $\lambda_{1} \rightarrow q_{1}  \in \Delta_{1}, \  \lambda_{2} \rightarrow q_{2}  \in  \Delta_{2}, \ \lambda_{3} \rightarrow q_{3}  \in  \Delta_{2}$ and $\lambda_{1} \sqsubseteq \lambda_{2}$ and $\lambda_{1} \sqsubseteq \lambda_{3}$, where $\lambda_{1} , \lambda_{2}, \lambda_{3} \in \Lambda$ (a)\\,
or\\
$f(q_{i_{1}}, \ldots, q_{i_{n}}) \rightarrow q_{1}  \in  \Delta_{1}, \ f(q_{i_{1}}', \ldots, q_{i_{n}}') \rightarrow q_{2}  \in \Delta_{2}, \  f(q_{i_{1}}'', \ldots, q_{i_{n}}'') \rightarrow q_{3}  \in \Delta_{3}$  and 
$\forall j \in [1,  n], \ (q_{i_{j}}, q_{i_{j}}') \in \mathcal{R}_{1}$ and $(q_{i_{j}}, q_{i_{j}}'') \in \mathcal{R}_{2}$, where $f \in \mathcal{F}_{n}$ (b)\\
\item $q_{1} \in \mathcal{Q}_{f_{1}} \Longleftrightarrow q_{2} \in \mathcal{Q}_{f_{2}}$ and $q_{1} \in \mathcal{Q}_{f_{1}} \Longleftrightarrow q_{3} \in \mathcal{Q}_{f_{3}}$ (c),
\end{itemize}
by definition of $\mathcal{R}_{1}$ and $\mathcal{R}_{2}$.\\

(a) We have $\lambda_{1} \rightarrow q_{1}  \in \Delta_{1}$, $\lambda_{2} \rightarrow q_{2}  \in  \Delta_{2}$ and $\lambda_{1} \sqsubseteq \lambda_{2}$. According to the minization algorithm, $\lambda_{2}$ is the least upper bound of all $\lambda \in \Lambda$ such that there exists $q \in \Q_{1}$ such that $\lambda \rightarrow q \in \Delta_{1}$ and $q$ is in the same equivalence classe as $q_{1}$ (i.e., $q \in [q_{1}]$ or $qPq_{1}$). Formally, $\lambda_{2} = \bigsqcup \lbrace \lambda \vert \lambda \rightarrow q \in \Delta_{1} \wedge qPq_{1}\rbrace$.

As $\A_{3}$ is minimized and contains $\mathcal{A}_{1}$, then $\lambda_{3}$ has to contain at least all the $\lambda \in \Lambda$ such that $\lambda \rightarrow q \in \Delta_{1}$ and $qPq_{1}$, or else $\mathcal{A}_{3}$ is not minimized.

So $\lambda_{3} \sqsupseteq \bigsqcup \lbrace \lambda \vert \lambda \rightarrow q \in \Delta_{1} \wedge qPq_{1}\rbrace$, so $\lambda_{2} \sqsubseteq \lambda_{3}$.\\

(b) We can immediately deduce that $\forall j \in [1,  n], \ (q_{i_{j}}', q_{i_{j}}'') \in \mathcal{R}$ by the definition of $\mathcal{R}$.\\

(c) So $q_{2} \in \mathcal{Q}_{f_{2}} \Longleftrightarrow q_{3} \in \mathcal{Q}_{f_{3}}$\\

And thanks to these properties deduced on $\mathcal{R}:\mathcal{Q}_{1}\times\mathcal{Q}_{2}$, we can deduce that $\mathcal{L}(\mathcal{A}_{2}) \subseteq \mathcal{L}(\mathcal{A}_{3})$.\\

As the least upper bound of two elements of a lattice is the best and
unique upper-approximation, this minimization algorithm returns the
best upper-approximation.$_{\square}$\\

\subsection{Refinement of the partition}  

In the previous paragraphs, the partition $\Pi$ was fixed. The
precision of the upper-approximations made during the determinization
algorithm depends on the finess of $\Pi$. For example, if $\Pi$ is of
size 1, all $\lambda$-transitions will be merged into one.

\begin{definition}[Refinement of a partition]

A partition $\Pi_{2}$ refines a partition $\Pi_{1}$ if :
$$\forall \pi_{2} \in \Pi_{2}, \ \exists \pi_{1} \in \Pi_{1}: \ \pi_{2} \sqsubseteq \pi_{1}$$

Let $\mathcal{A}_{1}$ = $\langle \Pi_{1}, \mathcal{Q}, \mathcal{F}, \mathcal{Q}_{f}, \Delta_{1} \rangle$ be a $P\LTA$. The $P\LTA$ $\mathcal{A}_{2}$ = $\langle \Pi_{2}, \mathcal{Q}, \mathcal{F}, \mathcal{Q}_{f}, \Delta_{2} \rangle$ refines $\mathcal{A}_{1}$ if :
\begin{enumerate}
\item $\Pi_{2}$ refines $\Pi_{1}$
\item the transitions of $\Delta_{2}$ are obtained by : $\forall \lambda \rightarrow q \in \Delta_{1}, \ \forall \pi_{2} \in \Pi_{2}, \ \lambda\sqcap\pi_{2} \rightarrow q \in \Delta_{2}$
\end{enumerate}

\end{definition}

Refining an automaton does not modify immediatly the recognized language, but leads to a more precise upper-approximation in the determinization, as illustrated herafter.

\begin{example}

Given $\Pi$ and $\Delta$ of example \ref{ex:part} and a partition $\Pi_{2} = \lbrace ]-\infty, -1[,[-1,0[, [0,0], ]0,+\infty[\rbrace$ that refines $\Pi$,
the set of transitions $\Delta_{2}$ of $P\LTA$ obtained with $\Pi_{2}$ is 
$\Delta_{2} = \lbrace[-3,-1[ \rightarrow q_{1}, [-1,-1] \rightarrow q_{1}, [-5,-2] \rightarrow q_{2}, [3,4] \rightarrow q_{3}, [-3, -1[ \rightarrow q_{4},
[-1,0[ \rightarrow q_{4}, [0,0] \rightarrow q_{4}, ]0,2] \rightarrow q_{4},
 f(q_{1},q_{2}) \rightarrow q_{5}, f(q_{3},q_{4}) \rightarrow q_{6}, f(q_{5},q_{6}) \rightarrow q_{f1}, f(q_{5},q_{6}) \rightarrow q_{f2}\rbrace$.

We now obtain this set of transitions for the deterministic corresponding $P\LTA$ with $\Pi_{2}$ :
$\Delta_{2_{d}} = \lbrace
[-5,-1[ \rightarrow q_{1,2,4}, [-1,0[ \rightarrow q_{1,4}, ]0,4] \rightarrow q_{3,4}, [0, 0] \rightarrow q_{4}, f(q_{1,2,4},q_{1,2,4}) \rightarrow q_{5},
f(q_{1,4},q_{1,2,4}) \rightarrow q_{5},
 f(q_{3,4},q_{3,4}) \rightarrow q_{6}, f(q_{3,4},q_{4}) \rightarrow q_{6}, f(q_{3,4},q_{1,2,4}) \rightarrow q_{6}, f(q_{3,4},q_{1,4}) \rightarrow q_{6}, f(q_{5},q_{6}) \rightarrow q_{f1,f2}\rbrace$.

\end{example}

%% % this section's title is Lattice Tree Automata, the name of the file is ope to avoir confusion with E.G. mainLTA.tex

\section{A Completion Algorithm for $\LTA$}
\label{sec:compl}

We are interested in computing the set of reachable states of an
infinite state system.  In general this set is neither representable
nor computable.  In this paper, we suggest to work within the Tree
Regular Model Checking framework for representing possibly infinite
sets of state. More precisely, we propose to represent configurations
by (built-in)terms and set of configurations (or set of states) by an
$\LTA$.

In addition, we assume that the behavior of the system can be
represented by conditional term rewriting systems ($TRS$), that are
term rewriting systems equipped with conjunction of conditions used to
restrain the applicability of the rule. Our conditional $TRS$, which
extends the classical definition of \cite{BaaderN-book98}, rewrites
terms defined on the concrete domain. This makes them independent from
the abstract lattice.  We first start with the definition of
predicates that allows us to express conditions on $TRS$.

\begin{definition}[Predicates]
\label{def:pred}
Let $\Pred$ be the set of predicates over $\mathcal{D}$. For instance
if $\rho$ is a $n$-ary predicate of $\Pred$ then $\rho:\mathcal{D}^n
\mapsto \{true,false\}$. We extend the domain of $\rho$ to $\TFX^n$ in
the following way:

$\rho(t_1, \ldots,t_n)= \left\{\begin{array}{l}
    \rho(u_1, \ldots, u_n) \mbox{ if } \forall i=1\ldots n: t_i \in \TFI \\
    \hspace*{1cm} \mbox{ where } \forall i=1\ldots n:  u_i=eval(t_i) \\
    false \mbox{ if } \exists j=1\ldots n: t_j \not \in \TFI
    \end{array}\right.$
\end{definition}
\noindent Observe that predicates are defined on built-in terms of the concrete domain. %As an example the predicate "$\geq$" : $1 + 3 \geq 2 \equiv true$.
If one of the predicate parameters cannot be evaluated into a built-in term, then the predicate returns false and the rule is not applied. 
%%\begin{example}
%%The following predicates $x < 2$, $2 > 3$, $x \geq y + 6$ $\in \Pred$ are eval%%uated on the domain of natural numbers.
%%\end{example}

%%We now formally define conditional term rewriting systems for (built-in) terms.

\begin{definition}[Conditional Term Rewriting System on $\mathcal{T}(\FNI \cup \FI,X)$] 
In our setting, a Term Rewriting System~(TRS) $\R$ is a set of {\em
  rewrite rules} $l \rw r \cond c_1 \wedge \ldots \wedge c_n$, where
$l\in \TFNIX$, $r \in \TFX = \mathcal{T}(\FNI \cup \FI,\X)$, $l \not
\in \X$, $\var(l) \supseteq \var(r)$ and $\forall i=1\ldots n: c_i=
\rho_i(t_1, \ldots, t_m)$ where $\rho_i$ is a $m$-ary predicate of
$\Pred$ and $\forall j=1\ldots m: t_j \in \TFIX \wedge \var(t_j)
\subseteq \var(l)$.
\end{definition}

\begin{example}
Using conditional rewriting rules, the factorial can be encoded as
follows:
\[\begin{array}{l}

fact(x) \rightarrow 1 \cond x\geq 0 \wedge x\leq 1 \\

fact(x) \rightarrow x * fact(x - 1) \cond x\geq 2

\end{array}\]

\end{example}

 The TRS $\R$ and the $\eval$ function induces a rewriting relation
 $\rw_{\R}$ on $\TF$ in the following way~: for all $s,t\in\TF$, we
 have $s \rwR t$ if there exist (1) a rewrite rule $l\rw r \cond
 c_1\wedge \ldots \wedge c_n \in \R$, (2) a position $p\in \pos(s)$,
 (3) a substitution $\sigma: \X\mapsto \TF$ such that $s|_p= l\sigma$,
 $t=\eval(s[r\sigma]_p)$ and $\forall i=1\ldots n: c_i
 \sigma=true$. The reflexive transitive closure of $\rwR$ is denoted
 by $\rw^{*}_{\R}$.\\

Our objective is to compute an $\LTA$ representing the set (or an
over-approximation of the set) of reachable states of an $\LTA$ $\A$
with respect to a TRS $\R$. In this paper, we adopt the completion
approach of ~\cite{GR10,FeuilladeGVTT-JAR04}, which intends to compute
a tree automaton $\aaex^k$ such that $\Lang(\aaex^k) \supseteq
\desc(\Lang(\A))$. The algorithm proceeds by computing the sequence of
automata $\aaex^{0}, \aaex^{1}, \aaex^{2},...$ that represents
successive applications of $\R$. Computing $\aaex^{i+1}$ from
$\aaex^{i}$ is called a {\em one-step completion}. In general the
sequence of automata may not converge in a finite amount of time. To
accelerate the convergence, we perform an abstraction operation which
accelerate the computation. Our abstraction relies on merging states
that are considered to be equivalent with respect to a certain
equivalence relation defined by a set of equations. We now give
details on the above constructions. Then, we show that, in order to be
correct, our procedure has to be combined with an evaluation that may
add new terms to the language of the automaton obtained by completion
or equational abstraction. We shall see that this closure property may
add an infinite number of transitions whose behavior is captured with
a new widening operator for $\LTA$.

\subsection{Computation of $\A_{i+1}$}
\label{sec:comp}

In our setting, $\aaex^{i+1}$ is built from $\aaex^{i}$ by using a
\textit{completion step} that relies on finding critical pairs.  Given
a substitution $\sigma:\X\mapsto\Q$ and a rule $l\rw r \cond c_1
\wedge \ldots \wedge c_n \in \R$, a critical pair is a pair
$(r\sigma', q)$ where $q\in\Q$ and $\sigma'$ is the greatest
substitution w.r.t $\leqTF$ such that $l\sigma \rw^*_{\aaex^i}q$,
$\sigma \geqTF \sigma'$ and $c_1\sigma' \wedge \ldots \wedge
c_n\sigma'$.

Observe that since both $\R$, $\aaex^i$, $\Q$ are finite, there is
only a finite number of such critical pairs.  For each critical pair
such that $r\sigma' \not\rw^*_{\aaex^i}q$, the algorithm adds two new
transitions $r\sigma' \rw q'$ and $q' \rightarrow q$ to $\aaex^i$.

%%$$
%%\xymatrix{
%% l\sigma \ar[r]_{\R}\ar[d]_{\A_i} & r\sigma \ar[d]^{\A_{i+1}}\\
%% q & \ar[l]^{\A_{i+1}} q'
%%}
%%$$

Building critical pairs for a rewriting rule $l\rightarrow r$ requires
to detect all substitutions $\sigma$ such that $l\sigma
\rightarrow^{*} q$, where $q$ is a state of the automaton. In what
follows, we use the standard {\em matching algorithm} introduced in
\cite{FeuilladeGVTT-JAR04}. This algorithm $Matching(l,\A,q)$, which
is described hereafter, matches a linear term $l$ with a state
$q$ in the automaton $\A$.  The solution returned by $Matching$ is a
disjunction of possible substitutions $\sigma_1\vee \ldots \vee
\sigma_n$ so that $l\sigma_i\rightarrow_{\A}^{*}q$.

Let us recall the standard matching algorithm:\\
\[\textrm{(Unfold) }
\dfrac {
  f(s_1, \dots, s_n)\unlhd f(q_1,\dots, q_n)
}{
  s_1 \unlhd q_1 \land \dots \land s_n \unlhd q_n
}\quad \quad
\textrm{(Clash) }
\dfrac {
  f(s_1, \dots, s_n) \unlhd g(q'_1, \dots q'_m)
}{
\bottom}
\]
\[\textrm{(Config) }
\dfrac {
  s \unlhd q
}{
  s \unlhd u_1 \lor \dots \lor s \unlhd u_k \lor \bottom
}
, \forall u_i,\ s.t.\ u_i \rightarrow q \in \Delta, \textrm{if }s \notin \X \textrm. 
\]

Moreover, after each application of one of these rules,
the result is also rewritten into disjunctive normal form, using:
\[
\dfrac{ 
  \phi_1\land (\phi_2 \lor \phi_3)
}{
  (\phi_1 \land \phi_2) \lor (\phi_1 \land \phi_3)
}\quad
\dfrac {\phi_1 \lor \bottom }{\phi_1}\quad
\dfrac {\phi_1 \land \bottom}{\bottom}
\]

However, as our
$TRS$ relies on conditions, we have to extend this matching algorithm
in order to guarantee that each substitution $\sigma_{i}$ that is a
solution of $l \rw r \Leftarrow c_1 \wedge \ldots \wedge c_n$
satisfies $c_1 \wedge \ldots \wedge c_n$. For example, given the rule
$f(x) \rightarrow f(g(x)) \Leftarrow x>3 \wedge x<7$ and the
transitions $[2,8] \rightarrow q_{1}$, $f(q_{1}) \rightarrow q_{2}$,
we have that the set of substitution returned by the matching
algorithm is $\{x \mapsto [2,8] \}$, which is restricted to $[3,7]$.

Restricting substitutions is done by a solver on abstract
domains. Such solver takes as input the lambda transitions of the
automaton and all conditions of the rules, and outputs a set of
substitutions of the form $\sigma'=\{x \mapsto \lambda_x, y \mapsto
\lambda_y\}$. Such solvers exist for various abstract domains (see
\cite{pichardie} for illustrations). In the present context, our
solver has to satisfy the following property:

\begin{property}[Correction of the solver]
\label{def:solve}
  Let $\sigma = \{x_1 \mapsto q_1, \dots, x_k \mapsto q_k\}$ be a
  substitution and $c= c_1 \wedge \dots \wedge c_n$ a conjunction of
  constraints. We consider $\sigma / c = \{ x_i \mapsto q_i \, | \,
  \exists 1\leq j\leq n, x_i \in \var(c_j)\}$ the restriction of the
  substitution to the constrained variables. We also define $S_c = \{ i
  \, | \, \exists 1\leq j\leq n, x_i \in \var(c_j)\}$.

For any tuple $\langle \lambda_i | i \in S_c \rangle$ such that
$\lambda_i \rw_\A^* q_i$, $\solve_\Lambda(\sigma/c, \langle \lambda_i
| i \in S_c \rangle, c)$ is a substitution $\sigma'$ such that (1) if
$i \not\in S_c$, $\sigma'(x_i) = q_i$, and (2) if $i \in S_c$,
$\sigma'(x_i) = \lambda'_i$. In addition, if a tuple of abstract
values $\langle \lambda''_i | i \in S_c \rangle$, satisfies (a)
$\forall i \in S_c$, $\lambda''_i \sqsubseteq \lambda_i$, and (b)
$\forall 1\leq j \leq n$, the substitution $\sigma''/c= \{x_i \mapsto
\lambda''_i \}$ satisfies $c_j$, then $\forall i \in S_c$,
$\lambda''_i \sqsubseteq \lambda'_i$.
\end{property}
Using Prop.\ref{def:solve}, the global function $\solve(\sigma,\A, c_1
\wedge \dots \wedge c_n)$ is defined as:
$$
\solve(\sigma,\A, c_1 \wedge \dots \wedge c_n) =
\mathop{\bigcup}_{\lambda_1 \rw_\A^* q_1,\dots,\lambda_k \rw_\A^* q_k}
\solve_\Lambda(\sigma/c, \langle \lambda_i | i \in S_c \rangle, c)$$

The following theorem ensures that $\solve(\sigma,\A, c_1 \wedge \dots
\wedge c_n)$ is an over-approximation of the solution of the
constraints.

\begin{theorem}
\label{th:solve}
$\solve(\sigma,\A, c_1 \wedge \dots \wedge c_n)$ is an over-approximation of the solutions of the constraints.
\end{theorem}

\begin{proof}
By Prop.\ref{def:solve}, we have that for any tuple $\langle \lambda_i
| i \in S_c \rangle$ such that $\lambda_i \rw_\A^* q_i$, then
$\solve_\Lambda(\sigma/c, \langle \lambda_i | i \in S_c \rangle, c)$
is a substitution $\sigma'$ such that if $i \in S_c$, $\sigma'(x_i) =
\lambda'_i$.  Let $\langle \lambda''_i | i \in S_c \rangle$ be a tuple
such that $\forall 1\leq j \leq n$, we have that the sustitution
$\sigma''/c= \{x_i \mapsto \lambda''_i \}$ satisfies $c_j$. Thanks to
Prop.\ref{def:solve}, we have that $\forall i \in S_c$, $\lambda''_i
\sqsubseteq \lambda'_i$. Since for all $i \in S_c$, $\lambda'_i$ is
returned by the solver, we can deduce that the set of substitutions
returned by the solver is an over-approximations of the solutions of
the constraints.
\end{proof}

Depending of the abstract domain, defining a solver that satisfies the
above property may be complex. However, we shall now see that an easy
and elegant solution can already be obtained for interval of
integers. As we shall see in Section \ref{sec:java}, such lattices act
as a powerful tool to simplify analysis of Java programs. Observe that
the algorithm for computing $\solve_\Lambda(\sigma/c, \langle
\lambda_i | i \in S_c \rangle, c)$ depends on the lattice $\Lambda$
and on the type of constraints of $c$. If $c$ is a conjunctions of
linear constraints and $\Lambda$ the lattice of intervals, the
algorithm computing $\solve_\Lambda(\sigma,\langle
\lambda_1,\dots,\lambda_k \rangle , c_1 \wedge \dots \wedge c_n)$ is:
\begin{enumerate}
\item $P_1$ is the convex polyhedron defined by the constraints $c_1 \wedge \dots \wedge c_n$,
\item $P_2$ is the box defined by the constraints $x_1 \in \lambda_1,\dots x_k \in \lambda_k$,
\item if $P_1 \sqcap P_2$, then we project $P_1 \sqcap P_2$ on each
  dimension~({\it i.e.} on each variable $x_k$) to obtain $k$ new
  intervals. Otherwise, $\solve_\Lambda(\sigma,\langle
  \lambda_1,\dots,\lambda_k \rangle , c_1 \wedge \dots \wedge c_n) =
  \emptyset$.
\end{enumerate}

\begin{definition}[Matching solutions of conditional rewrite rules]
\label{def:matchcond}
Let $\A$ be a tree automaton, $rl=l \rw r \cond c_1 \wedge \ldots
\wedge c_n$ a rewrite rule and $q$ a state of $\A$. The set of all
possible substitutions for the rewrite rule $rl$ is $\Omega(\A,rl,q)=
\{\sigma' \sep \sigma\in\matching(l,\A,q) \wedge \sigma'\in
\solve(\sigma,\A,c_1\wedge \ldots\wedge c_n) \wedge \nexists \sigma'':
r\sigma' \sqsubseteq r\sigma'' \rwA^* q\}$.
\end{definition}

Once the set of all possible restricted substitutions $\sigma_{i}$ has
been obtained, we have to add the rules $r\sigma_{i} \rightarrow^{*}
q$ in the automaton.  However, the transition $r \sigma_{i} \rw q$ is
not necessarily a normalized transition of the form $f(q_1, \ldots,
q_n) \rw q$, which means that it has to be normalized
first. Normalization is defined by the following algorithm.

\begin{definition}[Normalization]
\label{def:normalization}
Let $s \in \TFQ$, $q \in \Q$, $\FI$ the set of concrete interpretable
symbols used in the $TRS$, and $\A = \langle \F, \Q, \Q_{f},\Delta
\rangle$ an $\LTA$, where $\ = \FI \cup \FNI^{\#}$, and $\alpha: \FI
\rightarrow \FI^{\#}$ the abstraction function.  A {\em new state} is
a state of $\Q$ not occurring in $\Delta$. $Norm(s \rightarrow^{*} q)$
returns the set of normalized transitions deduced from $s$. $Norm(s
\rightarrow^{*} q)$ is inductively defined by:

\begin{enumerate}
\item if $s \in \FI^{0}$ (i.e., in the concrete domain used in rewrite
  rules), $Norm(s \rightarrow^{*} q) = \{\alpha(s) \rightarrow q\}$.
\item if $s \in \FNI^{0} \cup \FI^{\#^{0}}$ then $Norm(s
  \rightarrow^{*} q) = \{s \rightarrow q\}$,
\item if $s = f(t_1,\ldots,t_n)$ where $f \in \FNI^{n} \cup \FI^{n}$,
  then $Norm(s \rightarrow^{*} q) = \{f(q_{1}', \ldots, q_{n}')
  \rightarrow q\} \cup Norm(t_{1} \rightarrow q_{1}') \cup \ldots \cup
  Norm(t_{n} \rightarrow q_{n}')$ where for $i=1\ldots n$, $q'_i$ is
  either:
\begin{itemize}
\item the right-hand side of a transition of $\Delta$ such that $t_{i}
  \rightarrow_{\Delta}^{*} q_{i}'$
\item or a new state, otherwise.
\end{itemize}
\end{enumerate}
\end{definition}
\noindent
Observe that the normalization algorithm always terminates. We
conclude by the formal characterization of the one step completion.

%\begin{definition}[One step completed automaton $\comp_{\R}(\A)$]
%Let $\R$ the $TRS$, $\mathcal{X}$ the set of variables of $\R$, $\aaex^{i} = \langle \F, \Q_{i}, \Q_{i_{f}},\Delta_{i} \rangle$, and $\aaex^{i+1} = \langle \F, \Q_{i+1}, \Q_{i+1_{f}},\Delta_{i+1} \rangle$ be two $\LTA$. We denote by $\comp_{\R}(\aaex^{i})$ the $\LTA$ obtained by one step completion on $\aaex^{i}$. We say that $\comp_{\R}(\aaex^{i}) = \aaex^{i+1}$ if and only if $\forall l \rightarrow r \Leftarrow c_{1}\wedge\ldots\wedge c_{n} \in \R$, $\forall \sigma \in \Omega(\A,rl,q)$, $\forall q,q' \in Q_{i}$ such that $l\sigma \rightarrow^{*} q \in \aaex^{i}$ and $r\sigma \nrightarrow^{*} q' \in \aaex^{i}$ and $q'$ be a new state $\notin \Q_{i}$,
%we have $\Q_{i+1} = \Q_{i} \cup \bigcup q'$, and
%$$\Delta_{i+1}=\Delta_{i} \cup \bigcup_{\forall l \rightarrow r \in \R, \forall \sigma'}Norm(r\sigma' \rightarrow^{*} q')$$ 
%\end{definition}

\begin{definition}[One step completed automaton $\comp_{\R}(\A)$]
\label{def:completion-one-step}
  Let $\A= \langle \F, \Q, \Q_f, \Delta \rangle$ be a tree automaton, 
  $\R$ be a left-linear TRS. 
  We denote by $\comp_{\R}(\A)$ the one step completed automaton $\comp_{\R}(\A)= \langle \F, \Q', \Q_f,
\Delta' \rangle$ where:

$$ \Delta' = \Delta \cup \bigcup_{l\rightarrow r\in \R,\: q\in \Q,\:
  \sigma \in \Omega(\A,l\rightarrow r,q)} Norm(r\sigma \rightarrow^{*}
q') \cup\{q' \rw q\} $$

\noindent
where $\Omega(\A,l\rightarrow r,q)$ is the set of all possible substitutions defined in Def.\ref{def:matchcond}, $q' \notin \Q$ a new state and $\Q'$ contains all the states
of $\Delta'$.

\end{definition}

\subsection{Equational Abstraction}
\vspace{-0.2cm}
As we already said, completion may not terminate. In order to enforce
termination of the process, we suggest to merge states according to a
set {\em approximation equations} $E$. An approximation equation is of
the form $u=v$, where $u,v\in\TFNIX$.  Let $\sigma: \X \mapsto \Q$ be
a substitution such that $u\sigma \rw_{\A_{\R}^{i+1}} q$, $v\sigma
\rw_{\A_{\R}^{i+1}} q'$ and $q\neq q'$, then we know that some terms
recognized by $q$ and $q'$ are equivalent modulo $E$. An
over-approximation of $\aaex^{i+1}$, which we denote $\aapprox^{i+1}$,
can be obtained by merging states $q$ and $q'$.
%%After iteratively applying
%%the merging step on $\aaex^{i+1}$, we have to apply an evaluation
%%step.

\begin{definition}[$merge$]
\label{merge}
  Let $\A= \langle \F, \Q, \Q_F, \Delta \rangle$ be an $\LTA$ and
  $q_1,q_2$ be two states of $\A$. We denote by $merge(\A,q_1, q_2)$ the tree
  automaton where each occurrence of $q_2$ is replaced by $q_1$.
\end{definition}

\textbf{Equations on interpretable terms.}  In what follows, we need
to extend approximation equations to built-in terms. Indeed, as
illustrated in the following example, approximation equations defined
on $\TFNIX$ are not powerful enough to ensure termination.

\begin{example}
\label{eq}
Let $f(x) \rightarrow f(x+1)$ be a rewrite rule, $\{ [1,1] \rightarrow
q_{1}, [2,2] \rightarrow q_{2}, f(q_{2}) \rightarrow q_{f}\}$ be
transitions of an $\LTA$, then successive completion and normalization
steps will add transitions $q_{2} + q_{1} \rightarrow q_{3}$, $q_{3} +
q_{1} \rightarrow q_{4}$, $q_{4} + q_{1} \rightarrow q_{5}$, \ldots,
$q_{i}+q_{1}\rightarrow q_{i+1}$, \ldots Unfortunately, as classical
equations do not work on terms with interpretable symbols, this
infinite behaviour cannot be captured.
\end{example}

We define a new type of equation which works on interpretable terms,
that are applied with conditions. Such equations have the form $u=v
\cond c_1 \wedge \ldots \wedge c_n$, where $u,v \in \mathcal{T}(\FNI
\cup \FI, \mathcal{X})$. 
%%To simplify the presentation, we assume that
%%for all $i \in [1,n]$ $c_{i}$ is a m-ary predefined predicate defined
%%on $\TFX \times \TF^{m-1}$. This means that we assume that predicates
%%cannot have more than one variable in parameters, i.e., cannot be of
%%the form $x > y$, where $x, y$ are two variables.
%%%%%%%%%%%%%%%%%%%%%%%%%%%% ?????????Conditions are more restricted than those we used in $TRS$. ??????????%%%%%%%%%%%%%%%
%%used to find
%%substitutions $\sigma$ such that $l\sigma \rightarrow^{*} q$, where
%%$l$ is a left hand side of a rewrite rule and $q$ a state of the
%%current automaton, because $l \in \mathcal{T}(\FNI, \mathcal{X})$.
We observe that we can almost use the same matching algorithm than for
completion. The first main difference is that we need to match a term
$t \in \mathcal{T}(\FNI \cup \FI, \mathcal{X})$ built on interpreted
symbols on terms of $\mathcal{T}(\FNI \cup \FI^{\#}, \mathcal{X})$
recognized by the LTA $\A$. The solution is to use the same matching
algorithm on $\alpha(t)$ and $\A$, {\em i.e}
$Matching(\alpha(t),\A,q)$. Contrary to the completion case, we do not
need to restrict the substitutions obtained by the matching algorithm
with respect to the constraints of the equation, but simply guarantee
that such constraints are satisfiable, i.e., $\solve(\sigma,\A, c_1
\wedge \dots \wedge c_n)\not= \emptyset$.

\begin{example}
Equation $x = x + 1 \cond x > 3$ can be used to merge states $q_{4}$
and $q_{5}$ in Ex. \ref{eq}.
\end{example}

\begin{theorem}
\label{th:eq}

Let $\A$ be an $\LTA$ and $E$ a set of equations. We denote by
$\simp^!_E$ the transformation of $\A$ by merging equivalent states
according to $E$. The
language of the resulting automaton $\A '$ such that $\A \simp^!_E \A
'$ is an over-approximation of the language of $\A$, i.e.,
$\Lang(\A)\subseteq \Lang(\A')$.
\end{theorem}

\begin{proof}
Let $\A$ and $\A '$ two automata and $E$ be a set of equations such
that $\A \simp^!_E \A '$. The set of transition of $\A'$ is the same
as $\A$ with states merged according to equivalence classes determined
by $E$.  For all $t \in \TFX$, for all states $q$ of $\A$, let $Q =
\{q_{1}, \ldots, q, \ldots, q_{n}$ an equivalence class determined by
$E$.  We have that $t \in \Lang(\A,q) \Rightarrow t
\rightarrow^{*}_{\A} q \Rightarrow t \rightarrow^{*}_{\A '} Q
\Rightarrow t \in \Lang(\A ',Q)$.
\end{proof}

%%%%%%%%%%%%%%%%%%%%%%%%%%%%%%%%%%%%%%%%%%%%%%%%%%%%%%%%%%%%%%%%%%%%%%%%%%%%%%%%%%%%%%%%%%%%%%%%%%%%%%%%%%%%%%%%%%%%%%

\subsection{Evaluation and Correctness}
%%\vspace{-0.2cm}
In this section, we formally define completion on $\LTA$ and its
correctness. We first start with the evaluation of an $\LTA$.

\paragraph{Evaluation of a Lattice Tree Automaton.}

We observe that any set of concrete terms that contains the term $1 +
2$ should also contains the term $3$. While, this canonical property
can be naturally assumed when building the initial set of states, it
may eventually be broken when performing a completion step or by
merging states. Indeed, let $f(x) \rightarrow f(x + 1)$ be a rewrite
rule and $\sigma:x \mapsto q_{2}$ a substitution, a completion step
applied on $\{q_{1} \rightarrow [1,1], q_{2} \rightarrow [2,3],
f(q_{2}) \rightarrow q_{f}\}$ will add the rule $f(q_3)\rightarrow
q_4$, $q_{2} + q_{1} \rightarrow q_{3}$, and $q_3\rightarrow
q_f$. Since the language recognized by $q_3$ contains the term $q_{2}
+ q_{1}$, it should also contain the term $[3,4]$. Evaluation of this
set of transitions will add the transition $[3,4] \rightarrow q_{3}$.
This is done by applying the $propag$ function.

\begin{definition}[$propag$]
\label{def:propag}
  \[propag(\Delta)= \left\{\begin{array}{l}
      \Delta \mbox{ if } \exists \lambda \rw q \in \Delta \wedge
      eval(f(\lambda_1, \ldots,\lambda_k)) \sqsubseteq \lambda
      \\ \Delta \cup \{eval(f(\lambda_1, \ldots, \lambda_k)) \rw q\}, otherwise.
      \end{array}\right.\]

$\forall f \in \FI^{\#^k}: \forall q, q_1,\ldots, q_k \in \Q:\forall \lambda_1,
  \ldots, \lambda_k \in \Lambda: f(q_1, \ldots, q_k) \rw q \in
  \Delta \:\wedge$ $\{\lambda_1 \rw^*_{\Delta} q_1, \ldots, \lambda_k
  \rw^*_{\Delta} q_k \} \subseteq \Delta$.
\end{definition}

Using $propag$, we can extend the $\eval$ function to sets of
transitions and to tree automata in the following way.

\begin{definition}[$eval$ on transitions and automata]
\label{def:eval}
 ~\par
Let $\mu X$ the least fix-point obtained by iterating $propag$.
\begin{itemize}
\item
$\eval(\Delta)= \mu X.propag(X)\cup \Delta$ and
\item 
$\eval(\aut)= \langle \F, \Q, \Q_f, \eval(\Delta)\rangle$
\end{itemize}
\end{definition}

\begin{example}
Let $\Delta = \lbrace [3,6] \rightarrow q_{1}, [2,8] \rightarrow
q_{2}, q_{1} + q_{2} \rightarrow q_{3}, f(q_{3}) \rightarrow
q_{f}\rbrace$, then $propag$ will evaluate the term $[3,6] + [2,8]$
contained in the transition $q_{1} + q_{2} \rightarrow q_{3}$, and add
the transition $[5,14] \rightarrow q_{3}$ to the automaton.
\end{example}

\begin{theorem}
\label{th:eval}
$\Lang(\A) \subseteq \Lang(\eval(\A))$
\end{theorem}

\begin{proof}
By definition of $propag$ (Def.\ref{def:propag}), we have that $propag(\Delta) = Delta \mbox{ if } \exists \lambda \rw q \in \Delta \wedge eval(\lambda_1
      \bullet \ldots \bullet \lambda_k) \sqsubseteq \lambda$ or $propag(\Delta) = \Delta \cup \{eval(\lambda_1 \bullet \ldots \bullet \lambda_k) \rw q$. 
      In each case, $\Delta \subseteq propag(\Delta)$.

By definition of $eval$ (Def.\ref{def:eval}), $\eval(\Delta)= \mu X.propag(X)\cup \Delta$. Since 
$\Delta \subseteq propag(\Delta)$, we have that $\Delta \subseteq eval(\Delta)$.  Then we can deduce that
$\Lang(\A) \subseteq \Lang(\eval(\A))$.
\end{proof}

Observe that the fixpoint computation may not terminate. Indeed,
consider $\Delta = \lbrace[3,6] \rightarrow q_{1}, [2,8] \rightarrow
q_{2}, q_{1} + q_{2} \rightarrow q_{2}\rbrace$. The first iteration of
the fixpoint will evaluate the term $[3,6] + [2,8]$ recognized by
$q_{1} + q_{2} \rightarrow q_{2}$, which adds the transition $[5,14]
\rightarrow q_{2}$. Since a new element is in the state $q_{2}$, the
second iteration will evaluate the term $[3,6] + [5,14]$ recognized by
the transition $q_{1} + q_{2} \rightarrow q_{2}$, and will add the
transition $[8,20] \rightarrow q_{2}$. The third iteration will
evaluate the term $[3,6] + [8,20]$ to $q_{2}$ and this pattern will be
repeated in further operations.  Since there will always be a new
element of the lattice that will be associated to $q_{2}$, the
computation of the evaluation will not terminate. It is thus necessary
to apply a widening operator $\nabla_\Lambda: \Lambda \times \Lambda
\mapsto \Lambda$ to force the computation of $propag$ to
terminate. For example, if we apply such a widening operator on the
example above, after 3 iterations of the $propag$ function, the
transitions: $[2,8] \rightarrow q_{2}$, $[5,14] \rightarrow q_{2}$,
$[8,20] \rightarrow q_{2}$ could be replaced by $[2,+\infty[
    \rightarrow q_{2}$.

\begin{definition}[Automaton completion for $\LTA$]
Let $\A$ be a tree automaton, $\R$ a TRS and $E$ a set of
equations. At a step $i$ of completion, we denote by $\aapprox^{i}$
the $\LTA$ such that $\A^{i}_{\R} \simp^!_E \aapprox^{i}$.
\begin{itemize}
\item $\aapprox^0= \A$,
\item Repeat $\aapprox^{n+1}=\A'$ with
  $\comp_{\R}(eval(\aapprox^n))\simp^!_E \A''$ and $eval(\A'') = A'$,
\item Until a fixpoint $\aapprox^{*}= \aapprox^{k}=\aapprox^{k+1}$
  (with $k \in \NN)$ is joint.
\end{itemize}
\end{definition}

\noindent
A running example is described in section \ref{ss:example}.

\begin{theorem}[Completeness]
  \label{completeness}
  Let $\R$ be a left-linear TRS, $\A$ be a tree automaton and $\nr$ be a set of
  linear equations. If completion terminates on $\aapprox^*$ then
$$\Lang(\aapprox^*)\supseteq \desc(\Lang(\A))$$
\end{theorem}

\begin{proof}
We first show that $\Lang(\aapprox^*)\supseteq \Lang(\A)$. By
definition, completion only adds transitions to $\A$. Hence, we
trivially have $\Lang(\A^{1}_{\R})\supseteq \Lang(\A)$. Thanks
to Theorem \ref{th:eq}, we also know that $\A^{1}_{\R,E}$, the transformation of $\A^{1}_{\R} $ by merging states equivalent w.r.t. $E$, is such that 
$\Lang(\A^{1}_{\R,E}) \supseteq \Lang(\A^{1}_{\R})$. Hence, by transitivity of
$\supseteq$, we know that $\Lang(\A^{1}_{\R,E}) \supseteq \Lang(\A)$. This can
be successively applied to $\A^{2}_{\R,E}, \A^{3}_{\R,E},
\A^{4}_{\R,E}, \ldots$ so that $\Lang(\aapprox^*)\supseteq \Lang(\A)$.
Now, the next step of the proof consists in showing that for all term
$s \in \Lang(\A)$ if $s \rightarrow_{\R} t$ then $t \in
\Lang(\aapprox^*)$. First, note that by definition of application of
$E$ final states are preserved, i.e. if $q$ is a final state in $\A$
then if $\A '$ is the automaton where $E$ are applied in $\A$ and $q$
has been renamed in $q'$, then $q'$ is a final state of $\A '$. Hence
it is enough to prove that for all term $s \in \Lang(\A, q)$ if $s
\rightarrow^{*}_{\R} t$ then $\exists q' : t \in \Lang(\aapprox^*,
q')$. Because of previous result saying that
$\Lang(\aapprox^*)\supseteq \Lang(\A)$, from $s \in \Lang(\A,q)$ we
obtain that there exists a state $q'$ such that $s \in
\Lang(\aapprox^*, q')$. We know that $s \rightarrow^{*}_{\R} t$ hence,
what we have to show is that $t \in \Lang(\aapprox^*, q')$. By
induction on the length of $\rightarrow^{*}_{\R}$, we obtain that:
\begin{itemize}

\item if length is zero then $s \rightarrow^{*}_{\R} s$ and we
  trivially have that $s \in \Lang(\aapprox^*, q')$.
\item assume now that the property is true for any rewriting
  derivation of length less or equal to $n$, we prove that the
  property remains valid for a derivation of length less or equal to
  $n + 1$. Assume that we have $s \rightarrow^{n}_{\R} s'
  \rightarrow_{\R} t$. Using induction hypothesis, we obtain that $s'
  \in \Lang(\aapprox^*, q')$. It remains to prove that $t \in
  \Lang(\aapprox^*, q')$ can be deduced from $s' \rightarrow_{\R} t$.
  Since $s' \rightarrow_{\R} t$, we know that there exist a rewrite
  rule $l\rightarrow r \cond  c_1 \wedge \ldots \wedge c_n$, a position $p$ and a substitution $\mu
  : \mathcal{X} \mapsto \TF$ such that $s' = s'[l\mu]_{p}
  \rightarrow_{\R} eval(s'[r\mu]_{p})=t$ and for all $i \in [1,n]$,
  $c_{i}\mu = true$. Since $s' \in \Lang(\aapprox^*, q')$,
  $s'[l\mu]_{p} \rightarrow^{*}_{\aapprox^*} q'$ and by definition of
  the langage of an $\LTA$, we get that there exists $s''$ such that
  $s' \sqsubseteq s''$ and $s'' \rightarrow^{*}_{\aapprox^*} q'$. We
  can deduce that $s''[l\mu]_{p} \rightarrow^{*}_{\aapprox^*} q'$ and
  by definition of tree automata derivation, that there exists a state
  $q''$ such that $l\mu \rightarrow^{*}_{\aapprox^*} q''$ and
  $s''[q'']_{p} \rightarrow^{*}_{\aapprox^*} q'$.  Let $Var(l) =
  \{x_{1}, \ldots, x_{n} \}$, $l = l[x_{1}, \ldots, x_{n}]$ and
  $t_{1}, \ldots, t_{n} \in \TF$ such that $\mu = \{ x_{1}\mapsto
  t_{1}, \ldots, x_{n} \mapsto t_{n}\}$.  Since $l\mu = l[t_{1},
    \ldots, t_{n}] \rightarrow^{*}_{\aapprox^*} q''$, we know that
  there exist states $q_{1}, \ldots, q_{n}$ such that $\forall i \in
  [1,n]$, $t_{i} \rightarrow^{*}_{\aapprox^*} q_{i}$ and $l[q_{1},
    \ldots, q_{n}] \rightarrow^{*}_{\aapprox^*} q''$.  Let $\sigma =
  \{x_{1}\mapsto q_{1}, \ldots, x_{n} \mapsto q_{n}\}$, we thus have
  that $l\sigma \rightarrow^{*}_{\aapprox^*} q''$. Since $\aapprox^*$
  is a fixpoint of completion, from $l\sigma
  \rightarrow^{*}_{\aapprox^*} q''$ and the fact that for all $i \in
             [1,n]$, $c_{i}\mu = true$, we can deduce that $r\sigma
             \rightarrow^{*}_{\aapprox^*} q''$.  Furthermore, since
             $\forall i \in [1,n]$, $t_{i}
             \rightarrow^{*}_{\aapprox^*} q_{i}$, then $r\mu
             \rightarrow^{*}_{\aapprox^*} q''$.  Since besides of this
             $s''[q'']_{p} \rightarrow^{*}_{\aapprox^*} q'$, we have
             that $s''[r\mu]_{p} \rightarrow^{*}_{\aapprox^*} q'$.
             Since $s'\sqsubseteq s''$, this means by definition that
             $eval(s')\sqsubseteq eval(s'')$. Finally, since
             $s''[r\mu]_{p} \rightarrow^{*}_{\aapprox^*} q'$ and
             $eval(s')\sqsubseteq eval(s'')$, we can deduce that $t =
             eval(s'[r\mu]_{p}) \rightarrow^{*}_{\aapprox^*} q'$,
             hence $t \in \Lang(\aapprox^*, q')$.

\end{itemize}
\end{proof}

\noindent
Observe that the reverse does not hold as widening in evaluation may
introduce over-approximations.

\begin{remark}
We have two infinite dimensions, due to the state space, and due to infinite domain. The infinite behaviour of the system is abstracted thanks to the equations, and all the infinite behaviours due to the operations on elements of the lattice is captured by the widening step included in the evaluation step. Indeed, if we have lambda transitions added at each completion step with increasing (or decreasing) elements of the lattice (for example $[0,2] \rightarrow q$, $[2,4] \rightarrow q$, $[4,6] \rightarrow q$, \ldots), we have to perform a widening (here $[0,+\infty[$) to ensure the terminaison of the computation. But an infinite increasing (or decreasing) sequence of lambda transitions is necessarily obtained from a predifined operation of the lattice used in the rewrite rules. For example, the increasing sequence described above is necessarily obtained from a rewrite rule of the form  $u(\ldots, $\textbf{x}$, \ldots) \rightarrow v(\ldots, $\textbf{x + 2}$, \ldots)$. If we have the matching $x \rightarrow q_{1}$, and the rule $[2,2] \rightarrow q_{2}$, then it will add the transition $q_{1} + q_{2} \rightarrow q_{3}$, and since this rewrite rule leads to an infinite behaviour (always adding 2), we would have an infinite sequence $q_{3} + q_{2} \rightarrow q_{4}$, $q_{4} + q_{2} \rightarrow q_{5}$, and so on. To solve this problem, it is necessary to use an equation of the form $x = x + 2$. Then, $q_{1}$ is merged to $q_{3}$ and we have a transition $q_{1} + q_{2} \rightarrow q_{1}$ with an infinite evaluation abstracted thanks to the widening step included in the evaluation step. To summarize, an infinite sequence of lambda transitions is necessarily obtained from an operation used in the rewriting system, and since the transitions of an $\LTA$ containing operations have to be evaluated, the infinite behavior is always solved during the evaluation step. We can observe this on the example described hereafter in \ref{ss:example}.
\end{remark}

\section{A running example}
\label{ss:example}
Let $\mathbb{N}$ be the concrete domain, the set of intervals on
$\mathbb{N}$ be the lattice, $\R = \{ f(x) \rightarrow cons(x, f(x +
1)) \cond x < 3_{(A)}, f(x) \rightarrow cons(x, f(x + 2)) \cond x >
2_{(B)}\}$ be the $TRS$, $\A_{0}$ the $\LTA$ representing the set of
initial configurations, with the following set of transitions :
$\Delta_{0} = \{ [1,2] \rightarrow q_{1}, f(q_{1}) \rightarrow
q_{2}\}$, and $E = \{ x = x + 2 \cond x > 5 \}$ the set of equations.
We decide to perform a widening after three steps.\\

\noindent
\textbf{First step of completion}\\
\textit{One step completed automaton:} we can apply the rewrite rule $(A)$ with the substitution $x \mapsto q_{1}$, and so 
add $Norm(cons(q_{1}, f(q_{1} + 1)) \rightarrow q_{2}')$ and $q_{2}' \rightarrow q_{2}$ to $\Delta_{1}$.

So we have $\Delta_{2} = \Delta_{1} \cup \{ cons(q_{1}, q_{3})
\rightarrow q_{2}', q_{2}' \rightarrow q_{2}, f(q_{4}) \rightarrow
q_{3}, q_{1} + q_{[1,1]} \rightarrow q_{4}, q_{[1,1]} \rightarrow
[1,1] \}$.\\ Since there is new transitions, we have to perform the
evaluation step : transition $q_{1} + q_{[1,1]} \rightarrow q_{4}$ can
be evaluated, so $eval(\Delta_{2}) = \Delta_{2} \cup \{ [2,3]
\rightarrow q_{4} \}$.\\ \textit{Abstraction by merging states
  according to equations:} we cannot apply the set of equations yet
because there is no state recognizing "$x + 2$" such that $x > 5$.\\

\noindent
\textbf{Second step of completion}\\ \textit{One step completed
  automaton:} we can apply the rewrite rules $(A)$ and $(B)$ with the
substitution $x \mapsto q_{4}$, but this will be restricted by the
solver. In fact, $(A)$ will be applied on $[2,2]$ (condition $x < 3$),
and $(B)$ will be applied on $[3,3]$. So $Norm(cons([2,2], f([2,2] +
1)) \rightarrow q_{3}')$, $Norm(cons([3,3], f([3,3] + 2)) \rightarrow
q_{3}')$ and $q_{3}' \rightarrow q_{3}$ will be add to
$eval(\Delta_{2})$.

So we have $\Delta_{3} = eval(\Delta_{2}) \cup \{ [2,2] \rightarrow
q_{[2,2]}, cons(q_{[2,2]}, q_{5}) \rightarrow q_{3}', q_{3}'
\rightarrow q_{3}, f(q_{6}) \rightarrow q_{5}, q_{[2,2]} + q_{[1,1]}
\rightarrow q_{6}, [3,3] \rightarrow q_{[3,3]}, cons(q_{[3,3]}, q_{7})
\rightarrow q_{3}', f(q_{8}) \rightarrow q_{7}, q_{[3,3]} + q_{[2,2]}
\rightarrow q_{8} \}$.\\ Evaluation step: $eval(\Delta_{2}) =
\Delta_{2} \cup \{ [3,3] \rightarrow q_{6}, [5,5] \rightarrow q_{8}
\}$.  And as long as $[3,3] \rightarrow q_{[3,3]}$ and $[3,3]
\rightarrow q_{6}$, we can merge states $q_{[3,3]}$ and
$q_{6}$.\\ \textit{Abstraction step:} we cannot apply the set of
equations yet.\\

\noindent
\textbf{Third step of completion}\\ \textit{One step completed
  automaton:} we can apply the rewrite rule $(B)$ with the
substitution $x \mapsto q_{8}$.  So $Norm(cons(q_{8}, f(q_{8} + 2))
\rightarrow q_{7}')$, and $q_{7}' \rightarrow q_{7}$ will be add to
$Merge(eval(\Delta_{3}), q_{[3,3]}, q_{6})$.

So we have $\Delta_{3} = Merge(eval(\Delta_{3}), q_{[3,3]}, q_{6})
\cup \{cons(q_{8}, q_{9}) \rightarrow q_{7}', q_{7}' \rightarrow
q_{7}, f(q_{10}) \rightarrow q_{9}, q_{8} + q_{[2,2]} \rightarrow
q_{10} \}$.\\ Evaluation step: $eval(\Delta_{3}) = \Delta_{3} \cup \{
[7,7] \rightarrow q_{10} \}$.\\ \textit{Abstraction step:} As long as
$q_{8} + q_{[2,2]} \rightarrow q_{10}$, $[5,5] \rightarrow q_{8}$ and
$\gamma([5,5]) > 4$, $q_{8}$ and $q_{10}$ are merged according to the
set of equations $E$.\\

\noindent
\textbf{Fourth step of completion}\\ Let us see the full automaton at
this step.  We have $Merge(eval(\Delta_{3}), q_{8}, q_{10})) = \{
[1,2] \rightarrow q_{1}, f(q_{1}) \rightarrow q_{2}, cons(q_{1},
q_{3}) \rightarrow q_{2}', q_{2}' \rightarrow q_{2}, f(q_{4})
\rightarrow q_{3}, q_{1} + q_{[1,1]} \rightarrow q_{4}, q_{[1,1]}
\rightarrow [1,1], [2,3] \rightarrow q_{4}, [2,2] \rightarrow
q_{[2,2]}, cons(q_{[2,2]}, q_{5}) \rightarrow q_{3}', q_{3}'
\rightarrow q_{3}, f(q_{6}) \rightarrow q_{5}, q_{[2,2]} + q_{[1,1]}
\rightarrow q_{6}, [3,3] \rightarrow q_{6}, cons(q_{6}, q_{7})
\rightarrow q_{3}', f(q_{8}) \rightarrow q_{7}, q_{6} + q_{[2,2]}
\rightarrow q_{8}, [5,5] \rightarrow q_{8}, cons(q_{8}, q_{9})
\rightarrow q_{7}', q_{7}' \rightarrow q_{7}, f(q_{8}) \rightarrow
q_{9}, q_{8} + q_{[2,2]} \rightarrow q_{8}, [7,7] \rightarrow q_{8}
\}$.  Since the transitions have been modified thanks to the
equations, we have to perform an evaluation step.  We can nottice that
evaluation of the transition $q_{8} + q_{[2,2]} \rightarrow q_{8}$ is
infinite. In fact, it will add $[7,7] \rightarrow q_{8}$, $[9,9]
\rightarrow q_{8}$, $[11,11] \rightarrow q_{8}$, \ldots, and so on. So
we have to perform widening, that is to say, replace all the
transitions $\lambda \rightarrow q_{8}$ by $[5, +\infty[ \rightarrow
    q_{8}$.\\ \textit{One step completed automaton:} Thanks to the
    widening performed at the previous evaluation step, no more rule
    has to be add in the current automaton.  We have a fixed-point which
    is an over-approximation of the set of reachable states, and the
    completion stops.

%%%%%%%%%%%%%%%%%%%%%%%%Contre exemple correction.%%%%%%%%%%%%%%%%%%%%%%%%%%

\section{On Improving the Verification of Java Programs by TRMC}
\label{sec:java}
We now show how our formalism can simplify the analysis of JAVA
programs. In \cite{BoichutGJL-RTA07}, the authors developed a tool
called Copster~\cite{copster}, to compile a Java {\tt .class} file
into a Term Rewriting System (TRS). The obtained TRS models exactly a
subset of the semantics\footnote{essentially basic types, arithmetic,
  object creation, field manipulation, virtual method invocation, as
  well as a subset of the String library.} of the Java Virtual Machine
(JVM) by rewriting a term representing the state of the
JVM~\cite{BoichutGJL-RTA07}. States are of the form {\tt
  IO(st,in,out)} where {\tt st} is a program state, {\tt in} is an
input stream and {\tt out} and output stream. A program state is a
term of the form {\tt state(f,fs,h,k)} where {\tt f} is current frame,
{\tt fs} is the stack of calling frames, {\tt h} a heap and {\tt k} a
static heap. A frame is a term of the form {\tt frame(m,pc,s,l)} where
{\tt m} is a fully qualified method name, {\tt pc} a program counter,
{\tt s} an operand stack and {\tt t} an array of local variables. The
frame stack is the call stack of the frame currently being executed:
{\tt f}. For a given progam point {\tt pc} in a given method {\tt m},
Copster build a {\tt xframe} term very similar to the original {\tt
  frame } term but with the current instruction explicitly stated, in
order to compute intermediate steps.

One of the major difficulties of this encoding is to capture and
handle the two-side infinite dimension that can arise in Java
programs. Indeed, in such models, infinite behaviors may be due to
unbounded calls to method and object creation, or simply because the
program is manipulating unbounded datas such as integer
variables. While multiple infinite behaviors can be over-approximated
with completion (just like $a^nb^n$ can be approximated by $a^*b^*$),
this may require to manipulate structure of large size. As an example,
in \cite{BoichutGJL-RTA07}, it was decided to encode the structure of
configurations in an efficient manner, integer variables being encoded
in Peano arithmetic. Not only that this choice has an impact on the
size of the automata used to encode sets of configurations, but also
each classical arithmetic operation may require the application of
several rules.

As an example, let us consider the simple arithmetic operation $"300 +
400"$. By using \cite{BoichutGJL-RTA07}, this operation is represented
by $xadd(succ^{300}(zero), succ^{400}(zero))$, which reduces to $5$
rewriting rules detailled hereafter that have
to be applied $300$ times:\\
$xadd(zero,zero) \rightarrow result(zero)$\\
$xadd(succ(var(a)),pred(var(b))) \rightarrow xadd(var(a),var(b))$\\
$xadd(pred(var(a)),succ(var(b))) \rightarrow xadd(var(a),var(b))$\\
$xadd(succ(var(a)),succ(var(b))) \rightarrow xadd(succ(succ(var(a))),var(b))$\\
$xadd(pred(var(a)),pred(var(b))) \rightarrow xadd(pred(pred(var(a))),var(b))$\\
$xadd(succ(var(a)),zero) \rightarrow result(succ(var(a)))$\\
$xadd(pred(var(a)),zero) \rightarrow result(pred(var(a)))$\\
$xadd(zero,succ(var(b))) \rightarrow result(succ(var(b)))$\\
$xadd(zero,pred(var(b))) \rightarrow result(pred(var(b)))$\\

 This means that if at the program point {\tt pc} of method {\tt
  m} there is a bytecode {\tt add} then we switch to a {\tt xframe} in
order to compute the addition, i.e. apply $frame(m, pc, s, l)
\rightarrow xframe(add, m, pc, s, l) $. To compute the result of the
addition of the two first elements of the stack, we have to apply the
rule $xframe(add,m, pc, stack(b(stack(a,s))), l) \rightarrow
xframe(xadd(a,b), m, pc, s, l)$. Once the result is computed thanks to
all the rewrite rules of $xadd$, we can compute the next operation of
{\tt m}, i.e. go to the next program point by applying
$xframe(result(x), m, pc, s, l) \rightarrow frame(m, next(pc),
stack(x,s), l)$.

The use of $\LTA$ can drastically simplify the above
operations. Indeed, in our framework, we can encode natural numbers
and operations directly in the alpabet of the automaton. In such
context, the series of application of the rewritting rules is replaced
by a one step evaluation. As an example, the rewrite rule
$xframe(add,m, pc, stack(b(stack(a,s))), l) \rightarrow
xframe(xadd(a,b), m, pc, s, l)$ and rules $xadd$ encoding addition can
be replaced by $xframe(add,m, pc, stack(b(stack(a,s))), l) \rightarrow
xframe(result(a + b), m, pc, s, l)$.  Evaluation step of $\LTA$
completion will compute the result of addition of $a + b$ and add the
resulting term to the language of the automaton.

Other operations such as ``if-then-else'' can also be drastically
simplified by using our formalism. 
Indeed, with Peano numbers the evaluation of the condition of the instruction "{\tt if}" requires several rules. As an example, the instruction "{\tt if a=b then go to the program point x}" is encoded by the term $ifEqint(x,a,b)$, 
and the following rules will be applied:\\
$ifEqint(x,zero,zero) \rightarrow ifXx(valtrue,x)$\\
$ifEqint(x,succ(a),pred(b)) \rightarrow ifXx(valfalse,x)$\\
$ifEqint(x,pred(a),succ(b)) \rightarrow ifXx(valfalse,x)$\\
$ifEqint(x,succ(a),succ(b)) \rightarrow ifEqint(x,a,b)$\\
$ifEqint(x,pred(a),pred(b)) \rightarrow ifEqint(x,a,b)$\\
$ifEqint(x,succ(a),zero) \rightarrow ifXx(valfalse,x)$\\
$ifEqint(x,pred(a),zero) \rightarrow ifXx(valfalse,x$)\\
$ifEqint(x,zero,succ(b)) \rightarrow ifXx(valfalse,x)$\\
$ifEqint(x,zero,pred(b)) \rightarrow ifXx(valfalse,x)$\\

Rules of this type will disappear with $\LTA$ because an equality between two elements is directly evaluated, and so are all the predefined predicates.\\

In Copster, if at the program point {\tt pc} of the method {\tt m} we have 
an "{\tt if}" where the condition is an equality between two elements, we switch to a {\tt xframe} where the operation to evaluate is an "{\tt if}" with a equality condition between the two first elements of the stack, and which go to a program point $x$ if the condition is true. Then we can apply the rule $xframe(ifACmpEq(x),m,pc,stack(b,stack(a,s)),l) \rightarrow 
   xframe(ifEqint(x,a,b),m,pc,s,l)$ which permits to compute the solution, i.e. calls the $ifEqint$ rules detailed above. 

According to the result returned by these rules, we will go at program point $x$ if the condition is true or else to the next program point.
This is modelised by the two following rules:\\ 
 $xframe(ifXx(valtrue,x),m,pc,s,l) \rightarrow frame(m,x,s,l)$\\
$xframe(ifXx(valfalse,x),m,pc,s,l) \rightarrow frame(m,next(pc),s,l)$\\
   
%   xIfACmpEq(x,valtrue,a,b) -> ifEqint(x,a,b)\\
%xIfACmpEq(x,valfalse,a,b) -> ifXx(valfalse,x)\\
%
%
%ifEqint(x,zero,zero) -> ifXx(valtrue,x)\\
%ifEqint(x,succ(a),pred(b)) -> ifXx(valfalse,x)\\
%ifEqint(x,pred(a),succ(b)) -> ifXx(valfalse,x)\\
%ifEqint(x,succ(a),succ(b)) -> ifEqint(x,a,b)\\
%ifEqint(x,pred(a),pred(b)) -> ifEqint(x,a,b)\\
%ifEqint(x,succ(a),zero) -> ifXx(valfalse,x)\\
%ifEqint(x,pred(a),zero) -> ifXx(valfalse,x)\\
%ifEqint(x,zero,succ(b)) -> ifXx(valfalse,x)\\
%ifEqint(x,zero,pred(b)) -> ifXx(valfalse,x)\\

In $\LTA$ completion, thanks to the fact that predicates are directly evaluated and that we have conditional rules, all this rules are replaced by the two following conditional rules:  $xframe(ifACmpEq(x),m,pc,stack(b,stack(a,s)),l) \rightarrow 
   frame(m,x,s,l) \cond a = b$ (if $a = b$ we go to program point p)\\
   $xframe(ifACmpEq(x),m,pc,stack(b,stack(a,s)),l) \rightarrow 
   frame(m,x,s,l) \cond a \neq b$ (if $a \neq b$ we go to next program point)

\section{Conclusion and Future work}

We have proposed $\LTA$, a new extension of tree automata for tree
regular model checking of infinite-state systems whose configurations
can be represented with interpreted terms. One of our main
contributions is the development of a new completion algorithm for
such automata. We also give strong arguments that our encoding can
drastically improve the verification of JAVA programs in a TRMC-like
environment. As a future work, we plan to implement the
simplifications of Section \ref{sec:java} in Copster and combine them
with abstraction refinement techniques.

%%\\ $xadd(zero,zero) \rightarrow
%%result(zero)$\\ $xadd(succ(var(a)),pred(var(b))) \rightarrow
%%xadd(var(a),var(b))$\\ $xadd(pred(var(a)),succ(var(b))) \rightarrow
%%xadd(var(a),var(b))$\\ $xadd(succ(var(a)),succ(var(b))) \rightarrow
%%xadd(succ(succ(var(a))),var(b))$\\ $xadd(pred(var(a)),pred(var(b)))
%%\rightarrow
%%xadd(pred(pred(var(a))),var(b))$\\ $xadd(succ(var(a)),zero)
%%\rightarrow result(succ(var(a)))$\\ $xadd(pred(var(a)),zero)
%%\rightarrow result(pred(var(a)))$\\ $xadd(zero,succ(var(b)))
%%\rightarrow result(succ(var(b)))$\\ $xadd(zero,pred(var(b)))
%%\rightarrow result(pred(var(b)))$\\
%%\noindent

\bibliography{LTA}

%\newpage
%
%\input{App}

\end{document}